\documentclass{ws-ijtaf}
\usepackage{amsmath, amsfonts, amssymb}
\usepackage{bm}
\usepackage{graphicx}
\usepackage{lscape}
\usepackage{setspace}
\usepackage{subcaption}
\usepackage{booktabs} 
\usepackage{multirow}
\usepackage[bottom]{footmisc}

\usepackage{natbib}
\def\BIBand{and}%

\usepackage[colorlinks=true,linkcolor=blue,citecolor=blue]{hyperref} 

\usepackage{algorithm}
\usepackage{algorithmicx}
\usepackage{algpseudocode}
\makeatletter
\newenvironment{breakablealgorithm}
{
	\begin{center}
		\refstepcounter{algorithm}
		\hrule height.8pt depth0pt \kern2pt
		\renewcommand{\caption}[2][\relax]{
			{\raggedright\textbf{\ALG@name~\thealgorithm} ##2\par}%
			\ifx\relax##1\relax 
			\addcontentsline{loa}{algorithm}{\protect\numberline{\thealgorithm}##2}%
			\else 
			\addcontentsline{loa}{algorithm}{\protect\numberline{\thealgorithm}##1}%
			\fi
			\kern2pt\hrule\kern2pt
		}
	}{
		\kern2pt\hrule\relax
	\end{center}
}
\makeatother

\newtheorem{prop}{proposition}
\begin{document}
\markboth{Authors' Names}
{Mean-Variance Portfolio Management\\with Functional Optimization (Paper's Title)}

\catchline{}{}{}{}{}

\title{Mean-Variance Portfolio Management\\with Functional Optimization}

\author{Ka Wai Tsang}
\address{\email{kwtsang@cuhk.edu.cn}}

\author{Zhaoyi He}
\address{\email{zhaoyihe@link.cuhk.edu.cn}}

\address{School of Science and Engineering, The Chinese University of Hong Kong, Shenzhen, China}
\maketitle


\begin{abstract}
This paper introduces a new functional optimization approach to portfolio optimization problems by treating the unknown weight vector as a function of past values instead of treating them as fixed unknown coefficients in the majority of studies. We first show that the optimal solution, in general, is not a constant function. We give the optimal conditions for a vector function to be the solution, and hence give the conditions for a plug-in solution (replacing the unknown mean and variance by certain estimates based on past values) to be optimal. After showing that the plug-in solutions are sub-optimal in general, we propose gradient-ascent algorithms to solve the functional optimization for mean-variance portfolio management with theorems for convergence provided. Simulations and empirical studies show that our approach can perform significantly better than the plug-in approach. 
\end{abstract}



\section{Introduction}

For a portfolio consisting of $p$ stocks (or assets) with weakly stationary log-returns $\mathbf{r}_t=(r_{1,t},\dots,r_{p,t})^T$, let $\bm{\mu}=E\mathbf{r}_t=(\mu_{1},\dots,\mu_{p})^T$ and $\mathbf{V}=E\mathbf{r}_t\mathbf{r}_t^T$ be the mean and the second moment of $\mathbf{r}_t$, and let $\mathbf{w}=(w_1,\dots,w_p)^T$ be the vector of weights of the portfolio's value invested in stock $i$, $i=1,\dots,p$, such that $\mathbf{w}^T\mathbf{1}=\sum_{i=1}^{p}w_i=1$, where $\mathbf{1}=(1,\dots,1)^T$. Suppose we have a set of past returns, $\mathbf{S}_n = \{\mathbf{r}_1,\dots,\mathbf{r}_n\}$, and we want to construct a portfolio $\mathbf{w}$ for $\mathbf{r}_{n+1}$ to maximize a certain given objective function $F$ consisting of expected return $E(\mathbf{w}^T\mathbf{r}_{n+1})$ and variance $E((\mathbf{w}^T\mathbf{r}_{n+1})^2)-E(\mathbf{w}^T\mathbf{r}_{n+1})^2$. That is, we want to solve the problem
\begin{equation}\label{eq1}
\max_{\mathbf{w}\in\Omega} F(E(\mathbf{w}^T\mathbf{r}_{n+1}),E((\mathbf{w}^T\mathbf{r}_{n+1})^2)),
\end{equation}
for some continuously differentiable $C^1$ function $F:\mathbb{R}^2\rightarrow\mathbb{R}$, and $\Omega$ is a functional space that contains functions $\mathbf{w}:\mathbf{S}_n\rightarrow\mathbb{R}^P$ satisfying certain constrains (e.g., $\mathbf{w}^T\mathbf{1}=1$). The main stream in the literature assumes $\mathbf{w}$ to be a constant vector. In such case, $\Omega$ is a vector space and the problem \eqref{eq1} is reduced to
\begin{equation}\label{eq2}
\max_{\mathbf{w}\in\Omega} F(\mathbf{w}^T\bm{\mu},\mathbf{w}^T\mathbf{V}\mathbf{w}).
\end{equation}
For  $F(\mathbf{w}^T\bm{\mu},\mathbf{w}^T\mathbf{V}\mathbf{w})= -\mathbf{w}^T\mathbf{V}\mathbf{w}+(\mathbf{w}^T\bm{\mu})^2=-\mathbf{w}^T\bm{\Sigma}\mathbf{w}$, where $\bm{\Sigma}$ is the covariance matrix of $\mathbf{r}_t$, and $\Omega=\{\mathbf{w}:\mathbf{w}^T\bm{\mu}=\mu_*,\mathbf{w}^T\mathbf{1}=1,\mathbf{w}\geq\mathbf{0}\}$, where $\mu_*$ is a given target value for the mean return of a portfolio, and $\mathbf{w}\geq\mathbf{0}$ means $w_i\geq0$ for $i=1,\dots,p$, problem \eqref{eq2} is the mean-variance portfolio optimization considered by \citet{markowitz52,markowitz59}. The theory developed from \eqref{eq2}, for which Harry Markowitz received the 1990 Nobel Price in Economics, provided the first systematic treatment of a dilemma of getting high profit versus reducing risk (measured by variance). The constraint $\mathbf{w}^T\bm{\mu}=\mu_*$ in Markowitz's mean-variance formulation can be included in the objective function by using a Lagrange multiplier $\lambda^{-1}$. The parameter $\lambda>0$ can be regarded as a risk aversion coefficient, which may be difficult to determine in practice. Alternatively, investors often prefer to construct portfolio to maximize the information ratio $(\mathbf{w}^T\bm{\mu}-\mu_0)/\sigma_e$, where $\mu_0$ is the expected return of a benchmark investment (e.g., S\&P500) and $\sigma_e^2=Var(\mathbf{w}^T\mathbf{r}_{n+1}-r_{0,n+1})$, where $r_{0,n+1}$ is the return of the benchmark, is the variance of the portfolio's excess return over the benchmark investment; see \citet{grinold2000active}. If the benchmark investment is putting money in a risk-free bank account with interest rate $\mu_0$, it is called the Sharpe ratio. If $r_{0,n+1}=\mu_0$ is assumed to be nonrandom, then $\sigma_e^2=Var(\mathbf{w}^T\mathbf{r}_{n+1})=\mathbf{w}^T\bm{\Sigma}\mathbf{w}$ and the objective function for maximizing the Sharpe ratio is $F(\mathbf{w}^T\bm{\mu},\mathbf{w}^T\mathbf{V}\mathbf{w})=(\mathbf{w}^T\bm{\mu}-\mu_0)/\sqrt{\mathbf{w}^T\mathbf{V}\mathbf{w}-(\mathbf{w}^T\bm{\mu})^2}$. Considering the first and second moment of the portfolio return for portfolio construction can be justified by normal distribution assumption of $\mathbf{r}_t$. The well-known Black-Scholes model for stock price $S_{i,t}$ of the $i^{th}$ stock at time $t$ assumes $dS_{i,t}/S_{i,t}=\theta_i dt+\sigma_i d B_t^{(i)}$, where $\{B_t^{(i)},t\geq0\}$ is standard Brownian motion. It implies that $S_{i,t+1} = S_{i,t} exp(\theta_i-\sigma_i^2/2+\sigma_i(B_{t+1}-B_t))$, and hence the log-returns $r_{i,t}=log(S_{i,t}/S_{i,t-1})$ follows      $\mathcal{N}(\theta_i-\sigma_i^2/2,\sigma_i^2)$ distributions independently. Therefore the portfolio return $\mathbf{w}^T\mathbf{r}_{n+1}=\sum_{i=1}^{n}w_ir_{i,n+1}$ is also normal, which can be fully described by its first and second moments, if $\mathbf{w}$ is a constant vector. Under normal assumption, other common risk measures like $100(1-\alpha)\%$  Value at Risk (VaR) of a long position, $VaR_L(\alpha)=\inf\{x:P(\mathbf{w}^T\mathbf{r}_{n+1}\leq x)\geq \alpha\}=\mathbf{w}^T\bm{\mu}-z_{1-\alpha}\sqrt{\mathbf{w}^T\bm{\Sigma}\mathbf{w}}$, and $100(1-\alpha)\%$ expected shortfall (ES) of a long position, $ES_L(\alpha)=-E(\mathbf{w}^T\mathbf{r}_{n+1}\mid\mathbf{w}^T\mathbf{r}_{n+1}\leq VaR_L(\alpha))=-\mathbf{w}^T\bm{\mu}+\phi(z_{1-\alpha})\sqrt{\mathbf{w}^T\bm{\Sigma}\mathbf{w}}/\alpha$, where $z_{1-\alpha}$ is the $(1-\alpha)^{th}$ quantile of $\mathcal{N}(0,1)$ and $\phi$ is the density function of standard normal, can also be expressed as a function in the form of  $F(\mathbf{w}^T\bm{\mu},\mathbf{w}^T\mathbf{V}\mathbf{w})$; see Section 12.2.1 in \citet{lai2008statistical}. Other risk measures that only consider the first and second moments of return distributions can be found in \citet{steinbach2001markowitz}.

In practice, $\bm{\mu}$, $\mathbf{V}$ and $\bm{\Sigma}$ are unknown, and thus problem \eqref{eq2} cannot be solved directly. One commonly used approach is to estimate $\bm{\mu}$ and $\bm{\Sigma}$ from the historical data $\mathbf{S}_n$. By assuming $\mathbf{r}_t$ are weakly stationary, it is natural to estimate $\bm{\mu}$ and $\bm{\Sigma}$ by the sample mean $\hat{\bm{\mu}}=\frac{1}{n}\sum_{t=1}^{n}\mathbf{r}_t$ and sample covariance $\hat{\bm{\Sigma}}=\frac{1}{n}\sum_{t=1}^{n}(\mathbf{r}_t-\hat{\bm{\mu}})(\mathbf{r}_t-\hat{\bm{\mu}})^T$. However, such $\bm{\Sigma}$ estimate can be very poor if the number of stocks in the portfolio is large. \citet{frankfurter1976performance} and \citet{jobson1980estimation} have found that portfolios constructed by solving \eqref{eq2} with $\bm{\mu}$ and $\bm{\Sigma}$ replaced by $\hat{\bm{\mu}}$ and $\hat{\bm{\Sigma}}$ can perform worse than the equally weighted portfolio. Better estimate can be achieved if there is extra information about $\mathbf{r}_t$. If we assume multifactor models that relate the $p$ stock returns $r_{i,t}$ to $k$ factors $f_{1,t},\dots,f_{k,t}$ in a regression model of the form $r_{i,t}= \alpha_i +(f_{1,t},\dots,f_{k,t})^T\bm{\beta}_i+\varepsilon_{i,t}$, where $\alpha_i$ and $\bm{\beta}_i$ are unknown coefficients and $\varepsilon_{i,t}$ is an error term with mean 0 and is uncorrelated with the factors, then $\bm{\Sigma}$ can be decomposed into two parts: a systematic part due to the variability of certain unobserved factors, and an idiosyncratic part coming from the errors $\varepsilon_{i,t}$, and the number of parameters to estimate can be greatly reduced for small $k$ comparing with $O(p^2)$ for $\hat{\bm{\Sigma}}$. Starting from the studies of \citet{sharpe1964capital} and \citet{lintner1965valuation} that developed the capital asset pricing model (CAPM), which is a single-factor ($k=1$) model using the difference between the return of a hypothetical market portfolio (approximated by an index fund such as S\&P500 in practice) and the risk-free interest rate, various multifactor models have been proposed; see \citet{chen1986economic}, \citet{fama1993common}, \citet{carhart1997persistence}, \citet{cooper2008asset} and \citet{artmann2012determinants}. Section 3.4.3 of \citet{lai2008statistical} also suggest that the factors can be estimated by principal component analysis (PCA) from the past values $\mathbf{S}_n$. In such case, the estimated $\bm{\mu}$ and $\bm{\Sigma}$ can be considered as functions of $\mathbf{S}_n$.

Let $\hat{\bm{\Sigma}}^S$ be the sample covariance based on $\mathbf{S}_n$ and $\hat{\bm{\Sigma}}^F$ be the covariance estimate based on a multifactor model. $\hat{\bm{\Sigma}}^F$ is a better estimate than $\hat{\bm{\Sigma}}^S$ if the assumed multifactor model is a good model for $\mathbf{r}_t$, and can be worse than $\hat{\bm{\Sigma}}^S$ if the model describes $\mathbf{r}_t$ poorly. To have a more robust estimate, \citet{ledoit2003improved, ledoit2004honey}  propose to estimate $\bm{\Sigma}$ by a convex combination of $\hat{\bm{\Sigma}}^S$ and $\hat{\bm{\Sigma}}^F$, $\hat{\bm{\Sigma}}^{LW}=\hat{\delta}\hat{\bm{\Sigma}}^F+(1-\hat{\delta})\hat{\bm{\Sigma}}^S$, where $\hat{\delta}$ is a constant giving more weight to $\hat{\bm{\Sigma}}^S$ if the number of samples is large. Such shinkage estimate idea is consistent with the Bayes estimators with conjugate family of prior distributions. \citet{black1990asset} propose a quasi-Bayesian approach, which involves investor's subjective insights, to estimate $\bm{\mu}$. It is discussed in greater detail in \citet{bevan98usingthe} and \citet{he1999intuition}.
Following their ideas, further modifications have been proposed in \citet{meucci2005risk}, \citet{fabozzi2007robust}, and \citet{meucci2010black}.

While the vast majority of studies in solving problem \eqref{eq1} have focused on constant $\mathbf{w}$, which leads to problem formulation \eqref{eq2}, \citet{lai2011mean} points out that ``the construction of efficient portfolios when $\bm{\mu}$ and $\bm{\Sigma}$ are unknown is more complicated than trying to estimate them as well as possible and then plugging the estimates into \eqref{eq2}." They notice that the set $\mathbf{S}_n$ of past returns are actually random, and thus the estimate $\hat{\bm{\mu}}$ and $\hat{\bm{\Sigma}}$ that depend on $\mathbf{S}_n$ are also random. Therefore the weight vector $\mathbf{w}$ should be considered as a random vector. For example, for \eqref{eq2} with $F(\mathbf{w}^T\bm{\mu},\mathbf{w}^T\mathbf{V}\mathbf{w})=-\mathbf{w}^T\bm{\Sigma}\mathbf{w}= -\mathbf{w}^T\mathbf{V}\mathbf{w}+(\mathbf{w}^T\bm{\mu})^2$, and $\Omega=\{\mathbf{w}^T\bm{\mu}=\mu_*,\mathbf{w}^T\mathbf{1}=1\}$, the solution is $\mathbf{w}_0=\{B\bm{\Sigma}^{-1}\mathbf{1}-A\bm{\Sigma}^{-1}\bm{\mu}+\mu_*(C\bm{\Sigma}^{-1}\bm{\mu}-A\bm{\Sigma}^{-1}\mathbf{1})\}/D$, where $A=\bm{\mu}^T\bm{\Sigma}^{-1}\mathbf{1}$, $B=\bm{\mu}^T\bm{\Sigma}^{-1}\bm{\mu}$, $C=\mathbf{1}^T\bm{\Sigma}^{-1}\mathbf{1}$, and $D=BC-A^2$. If $\hat{\bm{\mu}}=\hat{\bm{\mu}}(\mathbf{S}_n)$ and $\hat{\bm{\Sigma}}=\hat{\bm{\Sigma}} (\mathbf{S}_n)=\hat{\mathbf{V}}(\mathbf{S}_n)-\hat{\bm{\mu}}(\mathbf{S}_n)\hat{\bm{\mu}}(\mathbf{S}_n)^T$ are used to replace $\bm{\mu}$ and $\bm{\Sigma}$ in $\mathbf{w}_0$, then $\mathbf{w}_0$ will be a function of $\mathbf{S}_n$ and hence $\mathbf{w}_0$ is random. The randomness of $\mathbf{w}$ makes problem \eqref{eq1} is not equivalent to \eqref{eq2}. Suppose $\mathbf{r}_t$ are i.i.d. with mean $\bm{\mu}$ and $\bm{\Sigma}$. By the law of iterated conditional expectations, $Var(\mathbf{w}^T\mathbf{r}_{n+1}) = E[Var(\mathbf{w}^T\mathbf{r}_{n+1}\mid\mathbf{S}_n)] + Var[E(\mathbf{w}^T\mathbf{r}_{n+1}\mid\mathbf{S}_n)]=E(\mathbf{w}^T\bm{\Sigma}\mathbf{w})+Var(\mathbf{w}^T\bm{\mu})$. While $\mathbf{w}^T\bm{\Sigma}\mathbf{w}=\mathbf{w}^T(\mathbf{V}-\bm{\mu}\bm{\mu}^T)\mathbf{w}$ for constant $\mathbf{w}$ in \eqref{eq2} can be considered as an estimate of $E(\mathbf{w}^T\bm{\Sigma}\mathbf{w})$ using the observed $\mathbf{S}_n$, the term $Var(\mathbf{w}^T\bm{\mu})$, which is nonzero if $\mathbf{w}$ is random, is omitted in \eqref{eq2}. \citet{lai2011mean} suggest that "this omission is an important root cause for the Markowitz optimization enigma related to 'plug-in' efficient frontiers". They consider problem formulation \eqref{eq1} with $F(E(\mathbf{w}^T\mathbf{r}_{n+1}),E((\mathbf{w}^T\mathbf{r}_{n+1})^2)) = E(\mathbf{w}^T\mathbf{r}_{n+1})-\lambda( E((\mathbf{w}^T\mathbf{r}_{n+1})^2)-(E(\mathbf{w}^T\mathbf{r}_{n+1}))^2)$ for some $\lambda > 0$, and treat $\bm{\mu}$ and $\bm{\Sigma}$ as state variables whose uncertainties are specified by their posterior distributions given the observations $\mathbf{S}_n$ in a Bayesian framework. Since problem \eqref{eq1} is not a standard stochastic optimization problem due to the higher order term $(E(\mathbf{w}^T\mathbf{r}_{n+1}))^2$, they first convert \eqref{eq1} to a standard stochastic optimization problem by showing that the maximizer of $E(\mathbf{w}^T\mathbf{r}_{n+1})-\lambda( E((\mathbf{w}^T\mathbf{r}_{n+1})^2)-(E(\mathbf{w}^T\mathbf{r}_{n+1}))^2)$ is the minimizer of $\lambda \mathbf{w}^T\mathbf{V}_n\mathbf{w} - \eta\mathbf{w}^T\bm{\mu}_n$ for some constant $\eta$, where $\bm{\mu}_n = E(\mathbf{r}_{n+1}\mid\mathbf{S}_n) $ and $\mathbf{V}_n = E(\mathbf{r}_{n+1}\mathbf{r}_{n+1}^T\mid\mathbf{S}_n)$ are the posterior mean and second moment matrix given $\mathbf{S}_n$. Without specifying particular prior distributions for $\bm{\mu}$ and $\bm{\Sigma}$, they propose using the empirical distribution of $\mathbf{S}_n=\{\mathbf{r}_1,\dots,\mathbf{r}_n\}$ to be the common distribution of the returns and setting $\bm{\mu}_n=\frac{1}{n}\sum_{t=1}^{n}\mathbf{r}_t$ and $\mathbf{V}_n=\frac{1}{n}\sum_{t=1}^{n}\mathbf{r}_t\mathbf{r}_t^T$. They call this approach for solving \eqref{eq1} as nonparametrical empirical Bayes (NPEB) approach.

In this paper, we agree the claim in \citet{lai2011mean} that the solution in \eqref{eq1} is not a fixed vector in general and treat $\mathbf{w}$ as a function of past returns $\mathbf{S}_n$. While \citet{lai2011mean} point out that the randomness of $\mathbf{w}$ comes from parameter estimation based on $\mathbf{S}_n$, Section \ref{sec2} show that the solution of \eqref{eq1} is not a constant function in general even if $\bm{\mu}$ and $\bm{\Sigma}$ are known. Functional optimization approaches are introduced to solve \eqref{eq1} for the case of $\Omega=\{\mathbf{w}:\mathbf{w}^T\mathbf{1}=1\}$ in Section \ref{sec3} and for the case of general closed and convex $\Omega$ in Section \ref{Omega}. Section \ref{sec5} gives the details of implementation of our proposed algorithms in practice. Simulations and an empirical study are given in Sections \ref{sec6} and \ref{sec7} to illustrate the performance of our algorithms. 

\section{Functional weighting maximizer}\label{sec2}
We will first focus on $\Omega = \{\mathbf{w}:\mathbf{w}^T\mathbf{1}=1\}$, which is corresponding to the case that short-selling is unlimitedly allowed. The case for more general constraint set $\Omega$ will be discussed in Section \ref{Omega}. Although $\bm{\mu}$ and $\bm{\Sigma}$ are commonly estimated under the assumption that $\mathbf{r}_t$ are independent, this assumption is overly restrictive as \citet{engle1986modelling} have noted volatility clustering and strong autocorrelations in the time series of squared returns, leading them to develop the ARCH and GARCH models in the 1980s. Therefore, we assume the distribution of $\mathbf{r}_{n+1}$ depends on its past values $\mathbf{S}_n$. Let $\bm{\mu}_n(\mathbf{S}_n)=E(\mathbf{r}_{n+1}\mid\mathbf{S}_n)$ and $\mathbf{V}_n(\mathbf{S}_n)=E(\mathbf{r}_{n+1}\mathbf{r}_{n+1}^T\mid\mathbf{S}_n)$ be the conditional mean and second moment of $\mathbf{r}_{n+1}$ respectively, then $E(\mathbf{w}^T\mathbf{r}_{n+1})=E(\mathbf{w}^T\bm{\mu}_n)$ and $E((\mathbf{w}^T\mathbf{r}_{n+1})^2)=E(\mathbf{w}^T\mathbf{V}_n\mathbf{w})$ as we assume $\mathbf{w}$ is a function of $\mathbf{S}_n$. Let $G(\mathbf{w})=F(E(\mathbf{w}^T\bm{\mu}_n),E(\mathbf{w}^T\mathbf{V}_n\mathbf{w}))$ and $\nabla G(\mathbf{w})=\nabla F(E(\mathbf{w}^T\bm{\mu}_n),E(\mathbf{w}^T\mathbf{V}_n\mathbf{w}))$ to simplify notations.

\subsection{Constant solution for independent returns}
If $\mathbf{r}_t$ are i.i.d.\ with mean $\bm{\mu}$ and covariance $\bm{\Sigma}$, then $\bm{\mu}_n(\mathbf{S}_n)=\bm{\mu}$ and $\mathbf{V}_n(\mathbf{S}_n)=\mathbf{V}=\bm{\Sigma}+\bm{\mu}\bm{\mu}^T$. Let $\mathbf{w}_B$ be the solution of \eqref{eq1}. If $\mathbf{w}_B$ is not a constant, let $\mathbf{w}_{B^\prime}=E(\mathbf{w}_B(\mathbf{S}_n))$, then we have $E(\mathbf{w}_{B^\prime}^T\bm{\mu})=E(\mathbf{w}_B^T\bm{\mu})$ and $E(\mathbf{w}_{B^\prime}^T\mathbf{V}\mathbf{w}_{B^\prime})-E(\mathbf{w}_B^T\mathbf{V}\mathbf{w}_B) =$ $ -E[(\mathbf{w}_B-\mathbf{w}_{B^\prime})^T\mathbf{V}(\mathbf{w}_B-\mathbf{w}_{B^\prime})]\leq 0 $, which implies that 
$Var(\mathbf{w}_{B^\prime}^T\mathbf{r}_{n+1})\leq Var(\mathbf{w}_B^T\mathbf{r}_{n+1})$. Therefore, if the value of the objective function $F$ increases when the expected return increases or the variance decreases, then we have $G(\mathbf{w}_{B^\prime})\geq G(\mathbf{w}_B)$ and hence $\mathbf{w}_B$ is a constant vector. This result, which is summarized in Lemma \ref{lem1}, justifies the use of problem formulation \eqref{eq2} for solving \eqref{eq1} when the returns are i.i.d. 

\begin{lemma}\label{lem1}
	If $\mathbf{r}_t$ are independent with mean $\bm{\mu}$ and covariance $\bm{\Sigma}$, and the objective function F in \eqref{eq1} is increasing in the first argument and is decreasing in the second argument, then the solution to the problem \eqref{eq1} is a constant vector function.
\end{lemma}

\subsection{Optimal condition for the solution}
Lemma \ref{lem1} supports plug-in approach to solve \eqref{eq1} (by solving \eqref{eq2}) when there are good estimates of $\bm{\mu}$ and $\bm{\Sigma}$ and $\mathbf{r}_t$ are independent. However it does not justify the use of plug-in approach when $\mathbf{r}_t$ is just weakly stationary. Let $\hat{\mathbf{w}}$ be the solution of \eqref{eq2} with $\bm{\mu}$ and $\mathbf{V}$ replaced by $\hat{\bm{\mu}}$ and $\hat{\mathbf{V}}$. Actually, although the formulation \eqref{eq2} is based on the assumption that $\mathbf{w}$ is constant, the fact that $\hat{\bm{\mu}}$ and $\hat{\mathbf{V}}$ are computed based on given $\mathbf{S}_n$ implies that $\hat{\bm{\mu}}$, $\hat{\mathbf{V}}$, and hence $\hat{\mathbf{w}}$ are functions of $\mathbf{S}_n$. We now derive the condition for the optimal solution for \eqref{eq1} and check when $\hat{\mathbf{w}}$ is equal to or close to the optimal solution.

For any $\mathbf{w}\in{\Omega}$ and $\bm{\delta}=\bm{\delta}(\mathbf{S}_n)$ satisfying $\bm{\delta}^T\mathbf{1}=0$, $\mathbf{w}+\bm{\delta}\in\Omega$. By Taylor expansion, we have 
\begin{equation}\label{eq4}
\begin{aligned}
G(\mathbf{w}+t\bm{\delta}) = G(\mathbf{w})+t(E(\bm{\delta}^T\bm{\mu}_n),2E(\mathbf{w}^T\mathbf{V}_n\bm{\delta}))\nabla G(\mathbf{w})+O(t^2),
\end{aligned}
\end{equation}
for small $t>0$. Let $\mathbf{w^*}$ be the solution of \eqref{eq1}, then we have $\lim\limits_{t\rightarrow 0^+}(G(\mathbf{w^*}+t\bm{\delta})-G(\mathbf{w^*}))/t\leq 0$, which implies $(E(\bm{\delta}^T\bm{\mu}_n), 2E(\bm{\delta}^T\mathbf{V}_n\mathbf{w^*}))\nabla G(\mathbf{w^*})\leq 0$. Since it holds for all $\bm{\delta}$ with $\bm{\delta}^T\mathbf{1}=0$, by replacing $\bm{\delta}$ by $- \bm{\delta}$, we get $E[\nabla G(\mathbf{w^*})^T(\bm{\mu}_n,$ $ 2\mathbf{V}_n\mathbf{w^*})^T\bm{\delta}]=0$. Note that $\nabla G(\mathbf{w^*})$ can be included in the expectation because $\nabla G(\mathbf{w^*})=\nabla F(E(\mathbf{w^*}^T\bm{\mu}_n), E(\mathbf{w^*}^T\mathbf{V}_n\mathbf{w^*})) $ $\in \mathbf{R}^2$ is a constant vector that does not depend on $\mathbf{S}_n$. On the other hand, if $F$ is a concave function around the point $p^*=(E((\mathbf{w^*})^T\bm{\mu}_n), E((\mathbf{w^*})^T\mathbf{V}_n\mathbf{w^*}))$ so that the Hessian matrix of $F$ is negative definite around $p^*$, then we have $G(\mathbf{w^*}+t\bm{\delta}) \leq G(\mathbf{w^*})+\nabla G(\mathbf{w^*})^T(0, t^2E(\bm{\delta}^T\mathbf{V}_n\bm{\delta}))$ if $E[\nabla G(\mathbf{w^*})^T(\bm{\mu}_n, 2\mathbf{V}_n\mathbf{w^*})^T\bm{\delta}]=0$ for all $\bm{\delta}$ satisfying $\bm{\delta}^T\mathbf{1}=0$. It then implies that $\lim\limits_{t\rightarrow 0^+} (G(\mathbf{w^*}+t\bm{\delta})-G(\mathbf{w^*}))/t\leq 0$ and the results in Lemma \ref{lem2}.
\begin{lemma}\label{lem2}
	The necessary condition for a function $\mathbf{w^*}=\mathbf{w^*}(\mathbf{S}_n)$ to be the solution of \eqref{eq1} is\begin{align*}
	\nabla G(\mathbf{w^*})^T E[(\bm{\mu}_n, 2\mathbf{V}_n\mathbf{w^*})^T\bm{\delta}]=0\quad \forall\bm{\delta} \in \{\bm{\delta}:\bm{\delta}^T\mathbf{1}=0\}.
	\end{align*}
	If F is a concave function around $p^*$, the condition is also sufficient.
	
\end{lemma}

By the method of Lagrange multiplier, the solution $\hat{\mathbf{w}}=\hat{\mathbf{w}}(\mathbf{S}_n)$ of \eqref{eq2} with $\Omega=\{\mathbf{w}:\mathbf{w}^T\mathbf{1}=1\}$ satisfies $
(\hat{\bm{\mu}},2\hat{\mathbf{V}}\hat{\mathbf{w}})\nabla F(\hat{\mathbf{w}}^T\hat{\bm{\mu}}, \hat{\mathbf{w}}^T\hat{\mathbf{V}}\hat{\mathbf{w}})-\hat{\lambda}\mathbf{1}=0$, $\forall\mathbf{S}_n$,
where $\hat{\lambda} = \nabla F(\hat{\mathbf{w}}^T\hat{\bm{\mu}},\hat{\mathbf{w}}^T\hat{\mathbf{V}}\hat{\mathbf{w}})^T(\mathbf{1}^T\hat{\mathbf{V}}^{-1}\hat{\bm{\mu}},2)/\mathbf{1}^T\hat{\mathbf{V}}^{-1}\mathbf{1}$. Take $\hat{\bm{\mu}}(\mathbf{S}_n)=\bm{\mu}_n(\mathbf{S}_n)$ and $\hat{\mathbf{V}}(\mathbf{S}_n)=\mathbf{V}_n(\mathbf{S}_n)$, it implies that plug-in solution $\hat{\mathbf{w}}(\mathbf{S}_n)$ satisfies \begin{equation}\label{eq6}
\begin{aligned}
E[\nabla F(\hat{\mathbf{w}}^T\bm{\mu}_n,
\hat{\mathbf{w}}^T\mathbf{V}_n\hat{\mathbf{w}})^T (\bm{\mu}_n, 2\mathbf{V}_n\hat{\mathbf{w}})^T\bm{\delta}]&=0
\end{aligned}
\end{equation}
for all $\bm{\delta} \in  \{\bm{\delta}^T\mathbf{1}=0\}$. Note that \eqref{eq6} does not imply the optimal condition in Lemma \ref{lem2} as $\nabla G(\mathbf{w^*})$, which is a constant vector, does not equal to $\nabla F(\hat{\mathbf{w}}(\mathbf{S}_n)^T\bm{\mu}_n(\mathbf{S}_n),
\hat{\mathbf{w}}(\mathbf{S}_n)^T\mathbf{V}_n(\mathbf{S}_n) $ $\hat{\mathbf{w}}(\mathbf{S}_n))$, which is a function vector of $\mathbf{S}_n$, in \eqref{eq6} in general. On the other hand, if $\nabla F$ (and hence $\nabla G$) is a constant vector, then \eqref{eq6} and Lemma \ref{lem2} implies that plug-in solution $\hat{\mathbf{w}}$ is also a solution to \eqref{eq1}. Let $U = E(\mathbf{w}^T\bm{\mu}_n)$ and $V=E(\mathbf{w}^T\mathbf{V}_n\mathbf{w})$. If we consider the problem formulation in \citet{lai2011mean}, we have \begin{equation}\label{eq7}
F(U,V) = U-\lambda(V-U^2) \qquad \nabla F(U,V) = (1+2\lambda U, -\lambda)^T.
\end{equation} 
If the variance of $\hat{\mathbf{w}}(\mathbf{S}_n)^T\bm{\mu}_n(\mathbf{S}_n)$ is small so that $\hat{\mathbf{w}}(\mathbf{S}_n)^T\bm{\mu}_n(\mathbf{S}_n)\simeq E(\hat{\mathbf{w}}^T\bm{\mu}_n)$ with high probability, then  $\nabla F(\hat{\mathbf{w}}(\mathbf{S}_n)^T\bm{\mu}_n(\mathbf{S}_n),
\hat{\mathbf{w}}(\mathbf{S}_n)^T\mathbf{V}_n(\mathbf{S}_n)\hat{\mathbf{w}}(\mathbf{S}_n))$ is close to $\nabla G(\hat{\mathbf{w}})$ in most cases by \eqref{eq7} and hence $\hat{\mathbf{w}}$ is a good approximate solution to \eqref{eq1} by Lemma \ref{lem2}. It supports the use of plug-in approach for the objective funcitons of the form \eqref{eq7} for some cases. Other objective functions that people would consider are 
\begin{align}
\{F,\nabla F\}&=\left\lbrace \frac{U-r_0}{\sqrt{V-U^2}},\quad (V-U^2)^{-3/2}(V-r_0U,\frac{r_0-U}{2})  \right\rbrace, \label{eq8}\\
\{F,\nabla F\}&=\left\lbrace  U -\lambda\sqrt{V-U^2},\quad (1+\frac{\lambda U}{\sqrt{V-U^2}},-\frac{\lambda}{2\sqrt{V-U^2}})  \right\rbrace. \label{eq9} 
\end{align}
Here the objective function in \eqref{eq8} is for Sharpe ratio with $r_0$ constant risk-free interest rate, and the one in \eqref{eq9} is for $100(1-\alpha)\%$ VaR under normal distributed returns assumption when $\lambda = z_{1-\alpha}$, and for $100(1-\alpha)\%$ ES when $\lambda=\phi(z_{1-\alpha})/\alpha$. 

\section{Functional optimization approach}\label{sec3}
As we know that the plug-in solution $\hat{\mathbf{w}}$ is not the solution of \eqref{eq1} in general, we want to find another function vector $\mathbf{w}$ such that $G(\mathbf{w})>G(\hat{\mathbf{w}})$. Since the solution of \eqref{eq1} is a constant vector if $\mathbf{r}_t$ are independent, we assume $\mathbf{r}_t$ depends on its past values. As choosing an appropriate time-series model for $\mathbf{r}_t$ is not the focus of this paper, we simply assume an AR(1) model for each stock's return,\begin{equation}\label{eq10}
r_{i,t}=\alpha_i+\beta_i r_{i,t-1}+\varepsilon_{i,t},\quad i=1,\dots,p,
\end{equation}
where $\alpha_i$ and $\beta_i$ are model coefficients and $\varepsilon_{i,t}$ is an error term with mean 0 independent of $r_{i,t-1}$, in this paper. Note that this assumed model can be replaced by any other models for the theorems and algorithms in the paper. For theoretical analysis, we assume $\bm{\mu}_n(\mathbf{S}_n)=E(\mathbf{r}_{n+1}\mid\mathbf{S}_n)=(\alpha_i+\beta_i r_{i,n})_{i=1,\dots,p}$ and $\mathbf{V}_n(\mathbf{S}_n)=E(\mathbf{r}_{n+1}\mathbf{r}_{n+1}^T\mid\mathbf{S}_n)=Var(\bm{\varepsilon}_{n+1}\mid\mathbf{S}_n)+\bm{\mu}_n(\mathbf{S}_n)\bm{\mu}_n(\mathbf{S}_n)^T$ are known functions of $\mathbf{S}_n$. In practice, we use the least squares estimators $\hat{\alpha}_i(\mathbf{S}_n)$ and $\hat{\beta}_i(\mathbf{S}_n)$ given $\mathbf{S}_n$ to replace $\alpha_i$ and $\beta_i$, and use $\hat{\bm{\mu}}(\mathbf{S}_n)=(\hat{\alpha}_i(\mathbf{S}_n)+\hat{\beta}_i(\mathbf{S}_n)r_{i,n})_{i=1,\dots,p}$ and $\hat{\mathbf{V}}(\mathbf{S}_n)=\frac{1}{n}\sum_{t=1}^{n}\hat{\bm{\varepsilon}}_t(\mathbf{S}_n)\hat{\bm{\varepsilon}}_t(\mathbf{S}_n)^T+\hat{\bm{\mu}}(\mathbf{S}_n)\hat{\bm{\mu}}(\mathbf{S}_n)^T$, where $\hat{\bm{\varepsilon}}_t(\mathbf{S}_n)=(r_{i,t}-\hat{\alpha}_i(\mathbf{S}_n)-\hat{\beta}_i(\mathbf{S}_n)r_{i,t-1})_{i=1,\dots,p}$, to replace $\bm{\mu}_n(\mathbf{S}_n)$ and $\mathbf{V}_n(\mathbf{S}_n)$.

\subsection{One-step algorithm}\label{onestep}
By equation \eqref{eq4}, $G(\hat{\mathbf{w}}+t\bm{\delta}) = G(\hat{\mathbf{w}})+ t \nabla G(\hat{\mathbf{w}})^T$ $ (E(\bm{\delta}^T\bm{\mu}_n), 2E(\hat{\mathbf{w}}^T\mathbf{V}_n\bm{\delta}))+O(t^2)$, which implies that if we choose $\bm{\delta}(\mathbf{S}_n)=(\bm{\mu}_n(\mathbf{S}_n),2\mathbf{V}_n(\mathbf{S}_n)\hat{\mathbf{w}}(\mathbf{S}_n))$, and $\nabla G(\hat{\mathbf{w}})$ and $\bm{\delta}(\mathbf{S}_n)$ does not equal to zero almost for sure, then the first order term of $\bm{\delta}$ in \eqref{eq4} becomes $E(||\bm{\delta}||^2)>0$ and hence $G(\hat{\mathbf{w}}+t\bm{\delta})>G(\hat{\mathbf{w}})$ for positive small $t$. Let $D_1(\mathbf{w})(\mathbf{S}_n)=(\bm{\mu}_n(\mathbf{S}_n), 2\mathbf{V}_n(\mathbf{S}_n)\mathbf{w}(\mathbf{S}_n))$ $\nabla G(\hat{\mathbf{w}})$. In order to satisfy the constraint $\bm{\delta}^T\mathbf{1}=0$ so that $\hat{\mathbf{w}}+t\bm{\delta} \in \Omega$, we consider $\bm{\delta}=\mathbf{P}D_1(\hat{\mathbf{w}})$, where $\mathbf{P}=(\mathbf{I}-\mathbf{1}\mathbf{1}^T/p)$ is the projection matrix projecting $D_1(\hat{\mathbf{w}})(\mathbf{S}_n)$ on to the space that is orthogonal to $\mathbf{1}$ for all $\mathbf{S}_n$. Note that $\mathbf{P}^T=\mathbf{P}$ and $\mathbf{P}^2=\mathbf{P}$, hence the first order term in \eqref{eq4} with this $\bm{\delta}$ satisfying $\bm{\delta}^T\mathbf{1}=0$ still equals to $E(||\bm{\delta}||^2)\geq0$. Note that $E(||\bm{\delta}||^2)=0$ if and only if $\mathbf{P}(\bm{\mu}_n(\mathbf{S}_n), 2\mathbf{V}_n(\mathbf{S}_n)\hat{\mathbf{w}}(\mathbf{S}_n))\nabla G(\hat{\mathbf{w}})=0$ for all $\mathbf{S}_n$ almost for sure, which implies that for all $\bm{\delta}_0$ satisfying $\bm{\delta}_0^T\mathbf{1}=0$, $\mathbf{P}^T\bm{\delta}_0=\bm{\delta}_0$ and $E[\nabla G(\hat{\mathbf{w}})^T(\bm{\mu}_n, 2\mathbf{V}_n\hat{\mathbf{w}})^T\bm{\delta}_0]=0$ and hence $\hat{\mathbf{w}}$ satisfies the optimal condition in Lemma \ref{lem2}. Therefore, before the optimal condition is satisfied, we have $E(||\bm{\delta}||^2)>0$ and $G(\hat{\mathbf{w}}+t\bm{\delta})>G(\hat{\mathbf{w}})$ for positive $t$ small enough. The computation of such $t$ and $\bm{\delta}$ function is presented in Algorithm \ref{onestepalgo}, which involves the computation of $E(\hat{\mathbf{w}}^T\bm{\mu}_n)$ and $E(\hat{\mathbf{w}}^T\mathbf{V}_n\hat{\mathbf{w}})$ for evaluating $G(\hat{\mathbf{w}})$ (and also $G(\hat{\mathbf{w}}+t\bm{\delta}))$. Since the expectations are very difficult, if not impossible, to compute explicitly in general, we use bootstrap method to generate $\mathbf{S}_n^{(1)},\dots,\mathbf{S}_n^{(b)}$ based on $\mathbf{S}_n$ for some fixed integer $B$, and then estimate the expected values by \begin{equation}\label{eq11}
\begin{aligned}
E(\hat{\mathbf{w}}^T\bm{\mu}_n) &\simeq \frac{1}{B} \sum_{b=1}^{B} \hat{\mathbf{w}}(\mathbf{S}_n^{(b)})^T\bm{\mu}_n(\mathbf{S}_n^{(b)}),\\ E(\hat{\mathbf{w}}^T\mathbf{V}_n\hat{\mathbf{w}}) &\simeq \frac{1}{B} \sum_{b=1}^{B}\hat{\mathbf{w}}(\mathbf{S}_n^{(b)})^T\mathbf{V}_n(\mathbf{S}_n^{(b)})\hat{\mathbf{w}}(\mathbf{S}_n^{(b)})
\end{aligned}
\end{equation}
If $\mathbf{r}_t$ are independent, then $\mathbf{S}_n^{(b)}$ can be generated by resampling $\{\mathbf{r}_1,\dots,\mathbf{r}_n\}$ in $\mathbf{S}_n$. However, as we assume $\mathbf{r}_t$ are dependent, we consider other resampling methods to handle the dependence. The details are given in Section \ref{Sn generation}.

\begin{breakablealgorithm}\label{onestepalgo}
	\caption{One-step algorithm}
	
	\begin{algorithmic}
		
		\Require
		$\mathbf{S}_n=\{\mathbf{r}_1, \dots,\mathbf{r}_n\}$, functions $\bm{\mu}_n$, $\mathbf{V}_n$ and $\hat{\mathbf{w}}$
		
		\State \textbf{Step 1} Generate $\mathbf{S}_n^{(1)},\dots,\mathbf{S}_n^{(B)}$ based on $\mathbf{S}_n$
		
		\State \textbf{Step 2} 
		\begin{enumerate}
			\item[] Compute $\bm{\mu}_n(\mathbf{S}_n^{(b)}), \mathbf{V}_n(\mathbf{S}_n^{(b)})$ and $\hat{\mathbf{w}}(\mathbf{S}_n^{(b)})$ for $b=1,\dots,B$.
			\item[] Compute $\hat{U} = \frac{1}{B}\sum_{b=1}^{B}\hat{\mathbf{w}}(\mathbf{S}_n^{(b)})^T\bm{\mu}_n(\mathbf{S}_n^{(b)})$, $\hat{V}=\frac{1}{B}\sum_{b=1}^{B}\hat{\mathbf{w}}(\mathbf{S}_n^{(b)})^T \mathbf{V}_n(\mathbf{S}_n^{(b)})$ $\hat{\mathbf{w}}(\mathbf{S}_n^{(b)})$.
		\end{enumerate}

		\State \textbf{Step 3} Compute  $\bm{\delta}_1(\mathbf{S}_n)=(\bm{\mu}_n(\mathbf{S}_n),\,2\mathbf{V}_n(\mathbf{S}_n)\hat{\mathbf{w}}(\mathbf{S}_n))$ $\nabla G(\hat{\mathbf{w}})$, where $\nabla G(\hat{\mathbf{w}})= \nabla F(\hat{U},\hat{V})$.
		
		\State \textbf{Step 4}\begin{enumerate}
			\item[] Compute $\bm{\delta}(\mathbf{S}_n^{(b)}) =\mathbf{P}\bm{\delta}_1(\mathbf{S}_n^{(b)})=(\mathbf{I}-\mathbf{1}\mathbf{1}^T/p)$ $\bm{\delta}_1(\mathbf{S}_n^{(b)})$ for $b=1,\dots,B$
			\item[] Compute $\hat{U}_\delta= \frac{1}{B}\sum_{b=1}^{B}\bm{\delta}(\mathbf{S}_n^{(b)})^T\bm{\mu}_n(\mathbf{S}_n^{(b)})$, $\hat{V}_\delta=\frac{1}{B}\sum_{b=1}^{B}2\bm{\delta}(\mathbf{S}_n^{(b)})^T \mathbf{V}_n(\mathbf{S}_n^{(b)})$ $\hat{\mathbf{w}}(\mathbf{S}_n^{(b)})$,
			$\hat{V}_{\delta \delta}=\frac{1}{B}\sum_{b=1}^{B}\bm{\delta}(\mathbf{S}_n^{(b)})^T\mathbf{V}_n(\mathbf{S}_n^{(b)})\bm{\delta}(\mathbf{S}_n^{(b)})$.
		\end{enumerate}
		
		\State \textbf{Step 5} Choose $\hat{t}$ such that $F(\hat{U}+\hat{t}\hat{U}_\delta,\hat{V}+\hat{t} \hat{V}_\delta+\hat{t}^2\hat{V}_{\delta \delta})>F(\hat{U},\hat{V})$
		
		\Ensure
		$\hat{t}$ and $\bm{\delta}$
	\end{algorithmic}
\end{breakablealgorithm}

\subsection{Multi-step algorithm}
Let $\mathbf{w}_0$ be an initial estimate of the solution $\mathbf{w^*}$ of \eqref{eq1}. It can be the plug-in solution $\hat{\mathbf{w}}$ or other trading strategies like equal weighting, i.e. $w_i=1/p$ for all $i$. Let $t_0$ and $\bm{\delta}_0=\mathbf{P}(\bm{\mu}_n,2\mathbf{V}_n\mathbf{w}_0) \nabla G(\mathbf{w}_0)$ be the step-size and direction computed by Algorithm \ref{onestepalgo} with $\hat{\mathbf{w}}=\mathbf{w}_0$. Then we can get an updated function\begin{equation}\label{eq13}
\begin{aligned}
\mathbf{w}_1(\mathbf{S}_n)=t_0a_0\mathbf{P}\bm{\mu}_n(\mathbf{S}_n)+(\mathbf{I}+2t_0b_0\mathbf{P}\mathbf{V}_n(\mathbf{S}_n))\mathbf{w}_0(\mathbf{S}_n),
\end{aligned}
\end{equation}
where $(a_0, b_0)^T=\nabla G(\mathbf{w}_0)$ do not depend on input $\mathbf{S}_n$. Setting $\hat{\mathbf{w}}=\mathbf{w}_1$ and applying Algorithm \ref{onestepalgo} again results in another estimate function $\mathbf{w}_2(\mathbf{S}_n)=\mathbf{w}_1(\mathbf{S}_n)+t_1\bm{\delta}_1(\mathbf{S}_n)$. Ignoring the bias from the expected values estimated by boostrap method, the argument in Section \ref{onestep} implies that $G(\mathbf{w}_2)>G(\mathbf{w}_1)>G(\mathbf{w}_0)$ and we can repeat Algorithm \ref{onestepalgo} multiple times to get a better estimate of $\mathbf{w^*}$.

Let $\mathbf{w}_k$ be the estimated function of $\mathbf{w^*}$ after applying Algorithm \ref{onestepalgo} $k$ times, then it is a function of $\mathbf{S}_n$ consisting of a sequence of functions, namely $\mathbf{w}_0,\mathbf{w}_1,\dots,\mathbf{w}_{k-1}$. Evaluating $\mathbf{w}_k$ at a particular $\mathbf{S}_n$ requires the evaluations of $\mathbf{w}_0,\dots,\mathbf{w}_{k-1}$ at $\mathbf{S}_n$. It means that the sequence of functions needs to be memorized, which is troublesome when $k$ is large. Fortunately, equation \eqref{eq13} gives us a solution to this issue. Let $B_k(\mathbf{S}_n)=\mathbf{I}+2t_kb_k\mathbf{P}\mathbf{V}_n(\mathbf{S}_n)$ be a function mapping $\mathbf{S}_n$ to $\mathbb{R}^{p\times p}$, where $t_k$ is the step-size chosen in the $k^{th}$ iteration of Algorithm \ref{onestepalgo} and $(a_k, b_k)^T=\nabla G(\mathbf{w}_k)$. Then similar to \eqref{eq13}, we have $\mathbf{w}_{k+1}(\mathbf{S}_n)=t_ka_k\mathbf{P}\bm{\mu}_n(\mathbf{S}_n)+B_k(\mathbf{S}_n)\mathbf{w}_k(\mathbf{S}_n)=[t_ka_k\mathbf{I}+\sum_{i=0}^{k-1}(\prod_{j=i+1}^{k}B_j(\mathbf{S}_n))t_i a_i]\mathbf{P}\bm{\mu}_n(\mathbf{S}_n)+(\prod_{j=0}^{k}B_j(\mathbf{S}_n)) \mathbf{w}_0(\mathbf{S}_n)$. It implies that the function $\mathbf{w}_k$ only involves $3$ functions ($\bm{\mu}_n$, $\mathbf{V}_n$ and $\mathbf{w}_0$), the sequence of chosen step-sizes $\{t_0,\dots,$ $t_{k-1}\}$ and the sequence of gradient vectors $\{\nabla G(\mathbf{w}_0),\nabla G(\mathbf{w}_1),\dots,$ $\nabla G(\mathbf{w}_{k-1})\}$, which are generated during the process of applying Algorithm \ref{onestepalgo} for getting $\{\mathbf{w}_0,\dots,$ $\mathbf{w}_k\}$. Note that we don't need to memorize the sequence of functions $\{\mathbf{w}_0,\dots,$ $\mathbf{w}_k\}$ in those $k$ times Algorithm \ref{onestepalgo} iterations. The details of computing $\{t_0,\dots,t_{k-1}\}$ and $\{\nabla G(\mathbf{w}_0),\dots, \nabla G(\mathbf{w}_{k-1})\}$ are presented in Algorithm \ref{multi-step}.

\begin{breakablealgorithm}\label{multi-step}
	\caption{Multi-step algorithm}
	
	\begin{algorithmic}
		
		\Require
		Observation $\mathbf{S}_n=\{\mathbf{r}_1, \dots,\mathbf{r}_n\}$, functions $\bm{\mu}_n$, $\mathbf{V}_n$ and $\mathbf{w}_0$
		
		\State \textbf{Step 1} Generate $\mathbf{S}_n^{(1)}, \dots,\mathbf{S}_n^{(B)}$ based on $\mathbf{S}_n$

		\State \textbf{Step 2} 
		\begin{enumerate}
			\item[] Compute $\bm{\mu}_n(\mathbf{S}_n^{(b)})$,$\mathbf{V}_n(\mathbf{S}_n^{(b)})$ and $\mathbf{w}_0(\mathbf{S}_n^{(b)})$ for $b=1,\dots,B$.
			\item[] 
			\begin{enumerate}
				\item[]Compute $U_0=\frac{1}{B}\sum_{b=1}^{B}\mathbf{w}_0(\mathbf{S}_n^{(b)})^T\bm{\mu}_n(\mathbf{S}_n^{(b)})$ $\approx E(\mathbf{w}_0^T\bm{\mu}_n)$,
				\item[]$\quad V_0=\frac{1}{B} \sum_{b=1}^{B}$ $\mathbf{w}_0(\mathbf{S}_n^{(b)})^T\mathbf{V}_n(\mathbf{S}_n^{(b)})\mathbf{w}_0(\mathbf{S}_n^{(b)})$ $\approx E(\mathbf{w}_0^T\mathbf{V}_n\mathbf{w}_0)$
			\end{enumerate}
		\end{enumerate}

		\State \textbf{Step 3} For $k = 0,1,2,\dots, K$
		\begin{enumerate}
			\item[3.1]  Compute $(a_k,b_k)^T = \nabla F(U_k,V_k)$
			\item[3.2] \begin{enumerate}
				\item[] Compute $\bm{\delta}_{Pk}(\mathbf{S}_n^{(b)})=\mathbf{P}(a_k\bm{\mu}_n(\mathbf{S}_n^{(b)})+2b_k$ $\mathbf{V}_n(\mathbf{S}_n^{(b)})\mathbf{w}_k(\mathbf{S}_n^{(b)}))$ for $b=1, \dots,B$,
				\item[] Compute $U_{\delta} = \frac{1}{B}\sum_{b=1}^{B}\bm{\delta}_{Pk}(\mathbf{S}_n^{(b)})^T\bm{\mu}_n(\mathbf{S}_n^{(b)})$, $V_{\delta} = \frac{1}{B}\sum_{b=1}^{B}(2\bm{\delta}_{Pk} (\mathbf{S}_n^{(b)})^T$ $\mathbf{V}_n(\mathbf{S}_n^{(b)})\mathbf{w}_k(\mathbf{S}_n^{(b)}))$, $V_{\delta\delta} = \frac{1}{B}\sum_{b=1}^{B}\bm{\delta}_{Pk}(\mathbf{S}_n^{(b)})^T\mathbf{V}_n(\mathbf{S}_n^{(b)}) \bm{\delta}_{Pk}(\mathbf{S}_n^{(b)})$
			\end{enumerate} 
			\item[3.3]
			\begin{enumerate}
				\item[] Choose $t_k$ such that $F(U_k+t_k U_{\delta},V_k+t_k V_{\delta}+t_k^2 V_{\delta\delta})>F(U_k,V_k)$
			\end{enumerate}
			
			\item[3.4] Update $U_{k+1}=U_k+t_k U_{\delta}$, $V_{k+1}= V_k+t_k V_{\delta}+t_k^2V_{\delta\delta}$, $\mathbf{w}_{k+1}(\mathbf{S}_n^{(b)})=\mathbf{w}_k(\mathbf{S}_n^{(b)})+t_k\bm{\delta}_k(\mathbf{S}_n^{(b)})$
		\end{enumerate}
		
		\Ensure $\mathbf{a}=(a_0,a_1, \dots,a_{k-1})$, $\mathbf{b}=(b_0,b_1, \dots,b_{k-1})$, and $\mathbf{t}=(t_0,t_1, \dots,t_{k-1})$
		
	\end{algorithmic}
\end{breakablealgorithm}

\subsection{Convergence rate}
If the solution $\mathbf{w^*}$ of \eqref{eq1} exists, we have seen that the estimated functions $\mathbf{w}_k$ having the property that the sequence $\{G(\mathbf{w}_k), k=0,1,2,\dots\}$ is increasing and hence converges to some point as the increasing sequence $\{G(\mathbf{w}_k)\}$ is bounded above by $G(\mathbf{w^*})$. In this section, we give the conditions for $G(\mathbf{w}_k)$ converging to $G(\mathbf{w^*})$ and the convergence rate when $\{\mathbf{w}_k\}$ converges. Given $\mathbf{w}$, let $\mathbf{z_{\mathbf{w}}}(\bm{\delta})=(E(\bm{\delta}^T\bm{\mu}_n),$ $ 2E(\mathbf{w}^T\mathbf{V}_n\bm{\delta}))$ and consider the Taylor expansion of $F$ up to second order derivative to get $G(\mathbf{w}+\bm{\delta})-G(\mathbf{w})$ equal to
\begin{equation}\label{eq15}
\begin{aligned}
E(D_1(\mathbf{w})^T\bm{\delta})+b(\mathbf{w})E(\bm{\delta}^T\mathbf{V}_n\bm{\delta})+\frac{1}{2}\mathbf{z_{\mathbf{w}}}^T\nabla^2G(\mathbf{w})\mathbf{z_{\mathbf{w}}}+\mathit{O}(E||\bm{\delta}||^3),
\end{aligned}
\end{equation}
where $b(\mathbf{w})$ is the second argument of $\nabla G(\mathbf{w})$ and $\nabla^2 G(\mathbf{w})$ is the Hessian matrix of $F$ evaluated at $(E(\mathbf{w}^T\bm{\mu}_n)$, $E(\mathbf{w}^T\mathbf{V}_n\mathbf{w}))$. Set $\mathbf{w}=\mathbf{w^*}$ in \eqref{eq15}, if we have $b(\mathbf{w^*})<0$ and $\nabla^2 G(\mathbf{w^*})$ is semi-negative definite ( or $b(\mathbf{w^*})\leq0$ and $\nabla^2 G(\mathbf{w^*})$  is negative definite), then the second order term of $\bm{\delta}$ in \eqref{eq15} is strictly negative. Suppose we further have the following strong concavity assumption around $\mathbf{w^*}$:
\begin{assumption}[Local strong concavity assumption]
	There exist $\varepsilon >0$ and $M>m>0$ such that for all $\bm{\delta}$ satisfying $\mathbf{w^*}+\bm{\delta}\in{\Omega}\cap B(\varepsilon)$, where $B(\varepsilon)=\{\mathbf{w}:E(||\mathbf{w}(\mathbf{S}_n)-\mathbf{w^*}(\mathbf{S}_n)||^2)\leq\varepsilon\}$, we have
	\begin{equation}\label{eq16}
	\begin{aligned}
	-\frac{m}{2} > \frac{b(\mathbf{w^*})E(\bm{\delta}^T\mathbf{V}_n\bm{\delta})+\mathbf{z_{\mathbf{w^*}}}^T\nabla^2 G(\mathbf{w^*}) \mathbf{z}_{\mathbf{w^*}}/2}{E||\bm{\delta}||^2}
	> -\frac{M}{2}.
	\end{aligned}
	\end{equation}
\end{assumption}
Therefore, for small $E||\bm{\delta}||$, we have $-\frac{m}{2}E||\bm{\delta}||^2\geq G(\mathbf{w^*}+\bm{\delta})-G(\mathbf{w^*}) -E(D_1(\mathbf{w^*})^T\bm{\delta})\geq -\frac{M}{2}E||\bm{\delta}||^2$ for $\bm{\delta}$ in the assumption. Based on the assumption \eqref{eq16}, we have the following theorem for convergence.

\begin{theorem}\label{convergence}
	Assume there exist $\varepsilon >0 $, $M>m>0$ such that 
	\begin{enumerate}
		\item[A1.] for all $\mathbf{w} \in B(\varepsilon)$ and $\mathbf{w}+\bm{\delta} \in B(\varepsilon)$, \begin{align*}
		\begin{aligned}
		-\frac{m}{2} \geq \frac{G(\mathbf{w}+\bm{\delta})-G(\mathbf{w}) -E(D_1(\mathbf{w})^T\bm{\delta})}{E||\bm{\delta}||^2} \geq -\frac{M}{2};
		\end{aligned}
		\end{align*}
		\item[A2.] the estimated functions $\mathbf{w}_k$, $k=0,1,2,\dots$, generated by Algorithm \ref{multi-step} with $t_k=\frac{2}{M+m}$ are all in $B(\varepsilon)$;
		\item[A3.] for $\mathbf{w} \in B(E||\mathbf{w}_0-\mathbf{w^*}||)$, $\mathbf{w}+\bm{\delta}(\mathbf{w})\in B(\varepsilon)$ and       $\mathbf{w}+\bm{\delta}(\mathbf{w})+\frac{1}{M-m}(D_1(\mathbf{w}+\bm{\delta}(\mathbf{w}))-D_1(\mathbf{w})+m\bm{\delta}(\mathbf{w}))\in B(\varepsilon)$, where 
		$\bm{\delta}(\mathbf{w})=\frac{1}{M+m}\mathbf{P}D_1(\mathbf{w})$.
	\end{enumerate}
	Then, we have
	\begin{equation}\label{eq17}
	E||\mathbf{w}_k-\mathbf{w^*}||\leq(\frac{M-m}{M+m})^kE||\mathbf{w}_0-\mathbf{w^*}||.
	\end{equation}
\end{theorem}
Assumption A1 in Theorem \ref{convergence} implies that the local strong concavity assumption \eqref{eq16} holds for $\mathbf{w}$ around $\mathbf{w^*}$. We have seen that the sequence $\{G(\mathbf{w}_k)\}$ converges to some value. Assumption A2 further assumes that $\{E||\mathbf{w}_k-\mathbf{w^*}||\}$ does not diverge so that $\mathbf{w}_k\in B(\varepsilon)$ for all $k$ if $\mathbf{w}_0$ is close to $\mathbf{w^*}$. Note that if there exists $M$ satisfies \eqref{eq16}, then $M$ can be chosen to be a very large number so that $(M+m)^{-1}$ can be any small positive number. In this sense. if $D_1(\mathbf{w})$ is a smooth bounded function for $\mathbf{w} \in B(\varepsilon)$, then $\bm{\delta}(\mathbf{w})$ in Assumption A3 is small in the sense that $E||\bm{\delta}||$ and $E||D_1(\mathbf{w}+\bm{\delta}(\mathbf{w}))-D_1(\mathbf{w})||$ are small enough for Assumption A3 holds. The proof of Theorem \ref{convergence} is presented in Appendix. The first inequality of the local strong concavity assumption \eqref{eq16} holds if the conditional second moment function $\mathbf{V}_n$ is strictly positive definite and the second argument $b(\mathbf{w^*})$ of $\nabla G(\mathbf{w^*})$, which represents the change of $F$ with respect to the second moment of $(\mathbf{w^*})^T\mathbf{r}_{n+1}$, is strictly negative. Therefore the assumption $b(\mathbf{w^*}) <0$ essentially means the value of the objective function $F$ increases when the variance of the optimal portfolio decrease, which is what we want in the design of $F$. For example, $b(\mathbf{w^*})$ in the problem formulation \eqref{eq7} and \eqref{eq9} are $-\lambda$ and $-\lambda/(2\sqrt{Var((\mathbf{w^*})^T\mathbf{r}_{n+1})})$ which are both strictly negative. For the formulation \eqref{eq8} for Sharpe ratio, $b(\mathbf{w^*})=\frac{1}{2}Var((\mathbf{w^*})^T\mathbf{r}_{n+1})^{-\frac{3}{2}}(r_0-E((\mathbf{w^*})^T\mathbf{r}_{n+1}))$ is strictly negative if the expected return is greater than the risk-free return, which is reasonable to assume to be true as we would then buy risk-free assets otherwise. Therefore, the local storng concavity assumption is reasonable for general smooth $F$ in practice.

\section{Closed and convex constraint set $\Omega$} \label{Omega}
In Section \ref{sec3}, we consider $\Omega=\{\mathbf{w}:\mathbf{w}^T\mathbf{1}=1\}$, which means short-selling is allowed unlimitedly. In practice, short-selling is limited and thus the constraint $\mathbf{w}\geq\ell$ (i.e., $w_i\geq\ell$ for $i=1,\dots,p$) for some non-positive $\ell$ should be added to $\Omega$. If $\ell=0$, it is the case that short-selling is not allowed. Note that $\Omega=\{\mathbf{w}\mid\mathbf{w}^T\mathbf{1}=1,\mathbf{w}\geq \ell\}$ is a closed convex set. For a general function $\mathbf{w}$ mapping $\mathbf{S}_n$ to $\mathbf{R}^p$, we define $P_{\Omega}\mathbf{w}$ to be a function in $\Omega$ in the way that, given any $\mathbf{S}_n=\{\mathbf{r}_1,\dots,\mathbf{r}_n\}$,\begin{equation}\label{eq18}
P_{\Omega}\mathbf{w}(\mathbf{S}_n)=\arg \min_{\bm{\mu}\in \Omega(\mathbf{S}_n)}||\bm{\mu}-\mathbf{w}(\mathbf{S}_n)||^2,
\end{equation}
where $\Omega(\mathbf{S}_n)=\{\mathbf{w}(\mathbf{S}_n):\mathbf{w}\in\Omega\}\subset \mathbb{R}^p$. As $\Omega$ is closed and convex, $\Omega(\mathbf{S}_n)$ is also closed and convex and $P_\Omega\mathbf{w}(\mathbf{S}_n)$ is an operation projecting $\mathbf{w}(\mathbf{S}_n)$ onto $\Omega(\mathbf{S}_n)$. By projection theorems on closed and convex sets (see, for example, \citet{cheney1959proximity}), we have the following properties for $P_\Omega\mathbf{w}(\mathbf{S}_n)$ for all $\mathbf{S}_n$ of past returns. \begin{equation}\label{eq19}
||P_\Omega\mathbf{w}_1(\mathbf{S}_n)-P_\Omega\mathbf{w}_2(\mathbf{S}_n)||\leq ||\mathbf{w}_1(\mathbf{S}_n)-\mathbf{w}_2(\mathbf{S}_n)||,
\end{equation}
\begin{equation}\label{eq20}
(\mathbf{w}(\mathbf{S}_n)-P_\Omega\mathbf{w}(\mathbf{S}_n))^T(\mathbf{w}_0(\mathbf{S}_n)-P_\Omega\mathbf{w}(\mathbf{S}_n))\leq 0, \qquad \forall \mathbf{w}_0 \in \Omega
\end{equation}
where $\mathbf{w}$, $\mathbf{w}_1$ and $\mathbf{w}_2$ are general functions mapping $\mathbf{S}_n$ to $\mathbb{R}^p$. We start from the ascent direction $D_1(\mathbf{w}_k)$ for $\mathbf{w}_k\in\Omega$ that is defined in Section \ref{onestep}. Instead of considering $\bm{\delta}_k=\mathbf{P}D_1(\mathbf{w}_k)$ so that $\mathbf{w}_{k+1}=\mathbf{w}_k+t\bm{\delta}_k$ satisfying $\mathbf{w}^T\mathbf{1}=1$, we consider $\mathbf{w}_{k+1}=P_\Omega\mathbf{y}_k$, where $\mathbf{y}_k=\mathbf{w}_k+tD_1(\mathbf{w}_k)$ for some $t>0$, so that $\mathbf{w}_{k+1}\in \Omega$. Let $\bm{\delta}_k=\mathbf{w}_{k+1}-\mathbf{w}_k=P_\Omega\mathbf{y}-P_\Omega\mathbf{w}_k$, then from \eqref{eq19}, $||\bm{\delta}_k(\mathbf{S}_n)||\leq||\mathbf{y}_k(\mathbf{S}_n)-\mathbf{w}_k(\mathbf{S}_n)||=t||D_1(\mathbf{w}_k)(\mathbf{S}_n)||$, and thus by equation \eqref{eq4} we have \begin{equation}\label{eq21}
G(\mathbf{w}_{k+1})= G(\mathbf{w}_k)+E(D_1(\mathbf{w}_k)^T(\mathbf{w}_{k+1}-\mathbf{w}_k))+O(t^2).
\end{equation}
Note that $E(D_1(\mathbf{w}_k)^T(\mathbf{w}_{k+1}-\mathbf{w}_k))$ is of $O(t)$ and we want to show that $E(D_1(\mathbf{w}_k)^T(\mathbf{w}_{k+1}-\mathbf{w}_k))>0$ for $t>0$. From \eqref{eq20}, for all $\mathbf{S}_n$,
\begin{equation}\label{eq22}
\begin{aligned}
&(\mathbf{y}_k(\mathbf{S}_n)-\mathbf{w}_k(\mathbf{S}_n))^T(\mathbf{w}_{k+1}(\mathbf{S}_n)-\mathbf{w}_k(\mathbf{S}_n))\\
=&||\mathbf{w}_{k+1}(\mathbf{S}_n)-\mathbf{w}_k(\mathbf{S}_n)||^2-(\mathbf{y}_k(\mathbf{S}_n)-\mathbf{w}_{k+1}(\mathbf{S}_n))^T(\mathbf{w}_k(\mathbf{S}_n)-\mathbf{w}_{k+1}(\mathbf{S}_n))\geq 0,
\end{aligned}
\end{equation}
which implies that $E(D_1(\mathbf{w}_k)^T(\mathbf{w}_{k+1}-\mathbf{w}_k))=\frac{1}{t}E((\mathbf{y}_k-\mathbf{w}_k)^T(\mathbf{w}_{k+1}-\mathbf{w}_k))>0$ and hence we have $G(\mathbf{w}_{k+1})>G(\mathbf{w}_k)$ if $\mathbf{w}_{k+1}\neq\mathbf{w}_k$. The modified Algorithm \ref{multi-step} for general closed and convex $\Omega$ is presented in Algorithm \ref{Omegaalgo}. Let $\mathbf{w^*}$ be the solution and take $\mathbf{w}_k=\mathbf{w^*}$ in \eqref{eq22}. Since $\mathbf{w^*}$ is optimal, the equal sign must holds in \eqref{eq22} and thus $\mathbf{w}_{k+1}=P_\Omega\mathbf{y}_k=\mathbf{w}_k=\mathbf{w^*}$, where $\mathbf{y}_k=\mathbf{w^*}+tD_1(\mathbf{w^*})$. In the case that $\Omega=\{\mathbf{w}:\mathbf{w}^T\mathbf{1}=1\}$, we can check that $P_\Omega\mathbf{y}_k=\mathbf{w^*}+t\mathbf{P}D_1(\mathbf{w^*})$, where $\mathbf{P}$ is the projection matrix in Algorithm \ref{onestepalgo}, and hence $\mathbf{P}D_1(\mathbf{w}_k)=0$, which leads to the optimal condition in Lemma \ref{lem2}. With this observation we can show that Theorem \ref{convergence} also holds for closed and convex $\Omega$. The proof is provided in the Appendix.
\begin{theorem} \label{thm2}
	With the same assumptions in Theorem 1, the inequality \eqref{eq17} also holds with closed and convex $\Omega$ and $\mathbf{w}_k$ generated by Algorithm \ref{Omegaalgo}.
\end{theorem}

\begin{breakablealgorithm}\label{Omegaalgo}
	\caption{Multi-step algorithm for closed and convex $\Omega$}
	
	\begin{algorithmic}
		
		\Require
		Observation $\mathbf{S}_n=\{\mathbf{r}_1, \dots,\mathbf{r}_n\}$, functions $\bm{\mu}_n$, $\mathbf{V}_n$ and $\mathbf{w}_0$
		
		\State \textbf{Step 1} Generate $\mathbf{S}_n^{(1)}, \dots,\mathbf{S}_n^{(B)}$ based on $\mathbf{S}_n$

		\State \textbf{Step 2} 
		\begin{enumerate}
			\item[] Compute $\bm{\mu}_n(\mathbf{S}_n^{(b)})$,$\mathbf{V}_n(\mathbf{S}_n^{(b)})$ and $\mathbf{w}_0(\mathbf{S}_n^{(b)})$ for $b=1,\dots,B$.
			\item[]
			\begin{enumerate}
				\item[]Compute $U_0=\frac{1}{B}\sum_{b=1}^{B}\mathbf{w}_0(\mathbf{S}_n^{(b)})^T\bm{\mu}_n(\mathbf{S}_n^{(b)})\approx E(\mathbf{w}_0^T\bm{\mu}_n)$,
				\item[]$\quad V_0=\frac{1}{B} \sum_{b=1}^{B}$ $\mathbf{w}_0(\mathbf{S}_n^{(b)})^T\mathbf{V}_n(\mathbf{S}_n^{(b)})\mathbf{w}_0(\mathbf{S}_n^{(b)})\approx E(\mathbf{w}_0^T\mathbf{V}_n\mathbf{w}_0)$
			\end{enumerate}
		\end{enumerate}

		\State \textbf{Step 3} For $k = 0,1,2,\dots,K$
		\begin{enumerate}
			\item[3.1]  Compute $(a_k,b_k)^T = \nabla F(U_k,V_k)$
			\item[3.2] \begin{enumerate}
				\item[] Compute $\bm{\delta}_k(\mathbf{S}_n^{(b)})=a_k\bm{\mu}_n(\mathbf{S}_n^{(b)})+2b_k\mathbf{V}_n(\mathbf{S}_n^{(b)})\mathbf{w}_k(\mathbf{S}_n^{(b)})$ for $b=1, \dots,B$,
				\item[]
				\begin{enumerate}
					\item[]Compute $U_{\delta} = \frac{1}{B}\sum_{b=1}^{B}\bm{\delta}_k(\mathbf{S}_n^{(b)})^T\bm{\mu}_n(\mathbf{S}_n^{(b)})$,
					\item[] $V_{\delta} = \frac{1}{B}\sum_{b=1}^{B}(2\bm{\delta}_k (\mathbf{S}_n^{(b)})^T$ $\mathbf{V}_n(\mathbf{S}_n^{(b)})\mathbf{w}_k(\mathbf{S}_n^{(b)}))$,
					\item[] $V_{\delta\delta} = \frac{1}{B}\sum_{b=1}^{B}\bm{\delta}_k(\mathbf{S}_n^{(b)})^T\mathbf{V}_n(\mathbf{S}_n^{(b)}) \bm{\delta}_k(\mathbf{S}_n^{(b)})$
				\end{enumerate} 
			\end{enumerate} 
			\item[3.3]
			\begin{enumerate}
				\item[] Choose $t_k$ such that $F(U_k+t_k U_{\delta},V_k+t_k V_{\delta}+t_k^2 V_{\delta\delta})>F(U_k,V_k)$
			\end{enumerate}
			
			\item[3.4] Update $U_{k+1}=U_k+t_k U_{\delta}$, $V_{k+1}= V_k+t_k V_{\delta}+t_k^2V_{\delta\delta}$, $\mathbf{y}_k=\mathbf{w}_k+t_k\bm{\delta}_k$, $\mathbf{w}_{k+1}(\mathbf{S}_n)=P_\Omega\mathbf{y}_k(\mathbf{S}_n)$
		\end{enumerate}
		\Ensure $\mathbf{a}=(a_0,a_1, \dots,a_{k-1})$, $\mathbf{b}=(b_0,b_1, \dots,b_{k-1})$, and $\mathbf{t}=(t_0,t_1, \dots,t_{k-1})$.
		
	\end{algorithmic}
\end{breakablealgorithm}

\section{Implementation in practice}\label{sec5}

The convergence in Theorems \ref{convergence} and \ref{thm2} requires the expected values in Algorithms \ref{multi-step} and \ref{Omegaalgo} to be exactly evaluated. However those expected values cannot be exactly computed in practice due to two issues. First, the true model for $\mathbf{r}_{n+1}$ and its past values $\mathbf{S}_n$ is unknown. The conditional expectations $\bm{\mu}_n(\mathbf{S}_n)=E(\mathbf{r}_{n+1}\mid\mathbf{S}_n)$ and $\mathbf{V}_n(\mathbf{S}_n)=E(\mathbf{r}_{n+1}\mathbf{r}^T_{n+1}\mid\mathbf{S}_n)$ can only be estimated by the observed samples $\mathbf{S}_n$ together with some prior knowledge about $\mathbf{r}_{n+1}$. Second, even if the model for $\mathbf{r}_{n+1}$ is known, the expected values of $\mathbf{w}_k(\mathbf{S}_n)^T\bm{\mu}_n(\mathbf{S}_n)$ and $\mathbf{w}_k(\mathbf{S}_n)^T\mathbf{V}_n(\mathbf{S}_n)\mathbf{w}_k(\mathbf{S}_n)$ can be very difficult to compute exactly since $\mathbf{w}_k$ can be a complicated function. The first issue can be addressed by conducting time series analyses on $\mathbf{r}_t$ to choose a good model and fit model parameters. Many tools have been developed for such time series analyses (e.g. see \citet{lai2008statistical}). In this paper, we simply consider the AR(1) model \eqref{eq10} to allow wrong model fitting to see the robustness of our approach in the simulations and the empirical study. For the second issue, we address it by using bootstrap resampling to estimate the expected values.

\subsection{$\mathbf{S}_n^{(b)}$ generation}\label{Sn generation}
The traditional bootstrap samples $\mathbf{S}_n^{(b)}= \{\mathbf{r}_1^{(b)},\dots,$ $\mathbf{r}_n^{(b)}\}$ by drawing with replacement from the observed sample $\mathbf{S}_n=\{\mathbf{r}_1,\dots,\mathbf{r}_n\}$. However, such resampling ignores the correlation between $\mathbf{r}_t$ and $\mathbf{r}_{t+h}$ for small lag $h\geq1$. To solve this issue, it is suggested to group $\{\mathbf{r}_1,\dots,\mathbf{r}_n\}$ into blocks that contain consecutive $\mathbf{r}_t$ and sample on those blocks so that certain level of correlation among the resampled returns is remained. This idea is called block bootstrap. \citet{hall1985resampling} and \citet{kunsch1989jackknife} introduce the block bootstrap as a nonparametric extension of the bootstrap for handling dependent data. Following their work, further studies and modifications of block bootstrap have been proposed; see \citet{buhlmann1997sieve}, \citet{paparoditis2002local} and \citet{lee2009double}. We apply the double block bootstrap proposed by \citet{lee2009double} for $\mathbf{S}_n^{(b)}$ generation in our simulations and the empirical study due to its simplicity. Another approach to handle dependent data is by parametric bootstrap. For example, if we know that $\mathbf{r}_t$ follow AR(1) model \eqref{eq10} with $\bm{\varepsilon}_t= (\varepsilon_{1,t},\dots,\varepsilon_{p,t})^T$ are independent, then we can apply traditional bootstrap to sample $\{\bm{\varepsilon}_1^{(b)},\dots,\bm{\varepsilon}_n^{(b)}\}$ by drawing with replacement from $\{\hat{\bm{\varepsilon}}_t(\mathbf{S}_n),t=1,\dots,n\}$ and generate $\mathbf{S}_n^{(b)} = \{\mathbf{r}_t^{(b)}=(r_{1,t}^{(b)},\dots,r_{p,t}^{(b)})^T, t=1,\dots,n \}$ with
\begin{equation}\label{eq12}
\begin{aligned}
r^{(b)}_{i,t}=\hat{\alpha}_i(\mathbf{S}_n)+\hat{\beta}_i(\mathbf{S}_n)r^{(b)}_{i,t-1}+\varepsilon^{(b)}_{i,t},
\end{aligned}
\end{equation}
where $\varepsilon^{(b)}_{i,t}$ is the $i^{th}$ entry of $\bm{\varepsilon}_t^{(b)}$ and $r^{(b)}_{i,0}=r_{i,0}=0$. In the case that $\bm{\varepsilon}_t$ are not independent, as pointed out in \citet{engle1986modelling}, we can apply the double block bootstrap to generate $\bm{\varepsilon}_t^{(b)}$ with $\mathbf{S}_n$ replaced by $\{\hat{\bm{\varepsilon}}_t(\mathbf{S}_n)\}$, and then generate $\mathbf{S}_n^{(b)}$ by \eqref{eq12}.

\section{Simulation studies}\label{sec6}
In this section, we illustrate the performance of Algorithm \ref{Omegaalgo} through simulations under various settings. For all the simulations, we choose the number $B$ of resampling to be $60$. Without assuming the true data generating mechanism of simulated returns $\mathbf{r}_t$, we use $\hat{\bm{\mu}}(\mathbf{S}_n)$ and $\hat{\mathbf{V}}(\mathbf{S}_n)$ that are based on the AR(1) model \eqref{eq10} as described in the first paragraph of Section \ref{sec3} to replace $\bm{\mu}_n(\mathbf{S}_n)$ and $\mathbf{V}_n(\mathbf{S}_n)$ in Algorithm \ref{Omegaalgo}. The initial function $\mathbf{w}_0$ is chosen to be the plug-in solution of problem \eqref{eq2} with $\bm{\mu}$ and $\mathbf{V}$ replaced by sample mean and sample second moment, which are common choices for weakly stationary returns in real applications. We generate returns $
r_{i,t}, i=1,\dots,20$, based on model \eqref{eq10} with different choices of $\alpha_i$, $\beta_i$ and $\varepsilon_{i,t}$:
\begin{enumerate}
	\item[(AR)] Assume $\varepsilon_{i,t}\sim\mathcal{N}(0,\sigma^2)$ independently. Set $\alpha_i=0.005$, $\beta_i=-0.4$, and $\sigma=0.04$ for all $i$. The values of the parameters are chosen by fitting the AR(1) model to the excess returns of Amazon (AMZN.O) over S\&P500, which is described in Section \ref{sec7}, in the period of January 2019 to December 2019. With these choices, the unconditional mean and variance of $r_{i,t}$ are $\alpha_i/(1-\beta_i)=0.0036$ and $\sigma^2=0.0016$.
	\item[(IID)] Assume $\varepsilon_{i,t}\sim\mathcal{N}(0,\sigma^2)$ independently. Set $\beta_i=0$ so that $r_{i,t}$ are independent. We choose $\alpha_i=0.0036$ so that the unconditional mean and variance are the same as in the Setting AR.
	\item[(GARCH)] Assume $\varepsilon_{i,t}=\sigma_{i,t}z_{i,t}$, where $\sigma_{i,t}^2=\gamma_0 +\gamma_1 \sigma_{i,t-1}^2+\gamma_2 \varepsilon_{i,t-1}^2$ and $z_{i,t} \sim \mathcal{N}(0,1)$ independently. It is the AR(1)-GARCH(1,1) model described in Section 6.3 of \citet{lai2008statistical}. Set $\alpha_i=0.005$, $\beta_i=-0.4$,$\gamma_1=\gamma_2=0.2$ and $\gamma_0=\sigma^2(1-(\gamma_1+\gamma_2))=0.00096$ so that the unconditional mean and variance are the same as in Settings AR and IID.
\end{enumerate}

The first 60 generated samples are training samples, and we construct portfolios $\mathbf{w}$ with portfolio returns $\mathbf{w}^T\mathbf{r}_t$, for $t=61,\dots,80$, so that the objective function $G(\mathbf{w})=F(E(\mathbf{w}^T\mathbf{r}_t),E((\mathbf{w}^T\mathbf{r}_t)^2))$ in \eqref{eq1} is maximum. We consider the mean-variance (MV) formulation \eqref{eq7}, the Sharpe ratio (SR) formulation \eqref{eq8} and the mean-standard-derivation (Msd) formulation \eqref{eq9} with $\lambda$ in \eqref{eq7} and \eqref{eq9} chosen to be $z_{0.9}$ or $0.1 z_{0.9}$. The constraint set is $\Omega=\{\mathbf{w}:\mathbf{w}^T\mathbf{1}=1,w_i\geq LB\}$ for $LB=-0.2$ or $LB=-1$. For portfolio $\mathbf{w}$ construction, we compare three approaches:
\begin{enumerate}
	\item Baseline approach. We set $w_i=1/20$ for all $i$. Note that the unconditional mean and variance for each $r_{i,t}$ are identical, and $r_{i,t}$ and $r_{j,t}$ are independent for $i\neq j$. Therefore, in the case of IID where the solution is constant, $E(\mathbf{w}^T\mathbf{r}_t)=\mathbf{w}^T E(\mathbf{r}_t)=\mathbf{w}^T\mathbf{1}E(r_{i,t})=0.0036$ does not depend on $\mathbf{w}$ and $Var(\mathbf{w}^T\mathbf{r}_t)=\mathbf{w}^T\Sigma\mathbf{w}=\sigma^2\sum_{i=1}^{20}w_i^2$ is minimum when $w_i=1/20$, which implies that this baseline approach leads to the true solution of \eqref{eq1} in IID case. Note that if we consider $\mathbf{w}$ as a constant vector and get to problem \eqref{eq2}, since $\bm{\mu}$ and $\mathbf{V}$ are the same in all settings, the baseline solution is also the solution of \eqref{eq2} in AR and GARCH settings. Let $r_t^{BL}=\sum_{i=1}^{20}r_{i,t}/20$ be the portfolio return for baseline approach at time $t=61,\dots,80$, and $\hat{G}^{BL}=F(\frac{1}{20}\sum_{t=1}^{20}r_{20+t}^{BL},\frac{1}{20}\sum_{t=1}^{20}(r_{20+t}^{BL})^2)$ as an estimate of $G(\mathbf{w}_{BL})$,where $\mathbf{w}_{BL}=\mathbf{1}/20$.
	\item Plug-in approach. For $61\leq t\leq 80$, we solve problem \eqref{eq2} with $\bm{\mu}$ and $\mathbf{V}$ replaced by the sample mean $\hat{\bm{\mu}}$ and sample second moment $\hat{\mathbf{V}}$ based on the set $\mathbf{S}_{t-1}=\{\mathbf{r}_{t-60},\dots,\mathbf{r}_{t-1}\}$. Let $\mathbf{w}_{PI}(\mathbf{S}_{t-1})$ be the solution at time $t$, $r_t^{PI}=\mathbf{w}_{PI}(\mathbf{S}_{t-1})^T\mathbf{r}_t$ and $\hat{G}^{PI}=F(\frac{1}{20}\sum_{t=1}^{20}r_{20+t}^{PI},\frac{1}{20}\sum_{t=1}^{20}(r_{20+t}^{PI})^2)$.
	\item Functional approach. We apply Algorithm \ref{Omegaalgo} on the training samples $\mathbf{r}_1,\dots,\mathbf{r}_{60}$ to get the outputs $\mathbf{a}$, $\mathbf{b}$ and $\mathbf{t}$. Then, for a given $t$, we compute $\mathbf{w}_{fun}(\mathbf{S}_{t-1})$ by \eqref{eq14} with $\mathbf{w}_0(\mathbf{S}_{t-1})=\mathbf{w}_{PI}(\mathbf{S}_{t-1})$. Let $r_t^{fun}=\mathbf{w}_{fun}(\mathbf{S}_{t-1})^T\mathbf{r}_t$ and $\hat{G}^{fun}=F(\frac{1}{20}\sum_{t=1}^{20}r_{20+t}^{fun},\frac{1}{20}\sum_{t=1}^{20}(r_{20+t}^{fun})^2)$.
\end{enumerate}
For each setting and each method, we repeat the simulations $100$ times independently to get $\{\hat{G}^{BL}(\ell),\hat{G}^{PI}(\ell)$, $\hat{G}^{fun}(\ell)\}$, $\ell=1,\dots,100$. Define $\Delta_{PI}=\frac{1}{100}\sum_{\ell=1}^{100}(\hat{G}^{PI}(\ell)-\hat{G}^{BL}(\ell))$ and $\Delta_{fun}=\frac{1}{100}\sum_{\ell=1}^{100}(\hat{G}^{fun}(\ell)-\hat{G}^{BL}(\ell))$, we test the null hypothesis $H_0:E(\Delta_{fun}) \leq E(\Delta_{PI})$ by applying one-sided pair t-test on $\{\hat{G}^{PI}(\ell),\hat{G}^{fun}(\ell)\}$, $\ell=1,\dots,100$. The p-values are presented in Table \ref{simuIID},\ref{simuAR} and \ref{simuGARCH}, which also report the number $n_+$ of times that $\hat{G}^{fun}(\ell)>\hat{G}^{PI}(\ell)$ and the number $n_0$ of times that $\hat{G}^{fun}(\ell) = \hat{G}^{PI}(\ell)$ in $100$ simulations.

\begin{table}[h]
	\caption{Comparison of the performance of Algorithm \ref{Omegaalgo} with plug-in and baseline approaches under IID setting. The values in parentheses under $\Delta_{PI}$ and $\Delta_{fun}$ are standard deviations.}
	\label{simuIID}
	\centering
	\small
	\resizebox{\textwidth}{!}{
		\begin{tabular}{cc|cccc|cccc}
			\hline
			 &    & \multicolumn{4}{c|}{LB=-0.2} & \multicolumn{4}{c}{LB=-1} \\
			$F$ & $\lambda$ & $\Delta_{PI}$ & $\Delta_{fun}$ & p-value & $(n_+,n_0)$ & $\Delta_{PI}$ & $\Delta_{fun}$ & p-value & $(n_+,n_0)$ \\
			\hline
			 SR &    & -0.2 (0.3) & -0.3 (0.3) & 1.00  & (26,0) & -0.2 (0.2) & -0.3 (0.3) & 1.00  & (32,0) \\
			\multirow{2}[0]{*}{MV} & $0.1z_{0.9}$ & -3e-3 (3e-2) & -3e-3 (3e-2) & 0.50  & (47,0) & -2e-2 (9e-2) & -2e-2 (9e-2) & 0.58  & (48,0) \\
			& $z_{0.9}$ & -1e-2 (2e-2) & -1e-2 (2e-2) & 0.99  & (37,0) & -5e-2 (5e-2) & -6e-2 (6e-2) & 1.00  & (31,0) \\
		    \multirow{2}[1]{*}{Msd} & $0.1z_{0.9}$ & -9e-3 (2e-2) & -1e-2 (2e-2) & 0.94  & (45,0) & -4e-2 (8e-2) & -4e-2 (7e-2) & 0.90  & (49,0) \\
			& $z_{0.9}$ & -5e-3 (4e-3) & -4e-2 (2e-2) & 1.00  & (0,0) & -5e-3 (3e-3) & -2e-2 (3e-2) & 1.00  & (0,84) \\
			\hline
	\end{tabular}
}%
\end{table}%

\begin{table}[h]
	\caption{Comparison of the performance of Algorithm \ref{Omegaalgo} with plug-in and baseline approaches under AR setting. The values in parentheses under $\Delta_{PI}$ and $\Delta_{fun}$ are standard deviations.}
	\label{simuAR}%
	\centering
	\small
	\resizebox{\textwidth}{!}{
		
		\begin{tabular}{cc|cccc|cccc}
			\hline
			 &    & \multicolumn{4}{c|}{LB=-0.2} & \multicolumn{4}{c}{LB=-1} \\
			$F$ & $\lambda$ & $\Delta_{PI}$ & $\Delta_{fun}$ & p-value & $(n_+,n_0)$ & $\Delta_{PI}$ & $\Delta_{fun}$ & p-value & $(n_+,n_0)$ \\
			\hline
			 SR &    & -0.3 (0.2) & 0.8 (0.4) & 2e-52 & (100,0) & -0.3 (0.2) & 1 (0.4) & 3e-56 & (100,0) \\
			\multirow{2}[0]{*}{MV} & $0.1z_{0.9}$ & -2e-2 (2e-2) & 3e-2 (2e-2) & 5e-55 & (100,0) & -9e-2 (8e-2) & 2e-2 (6e-2) & 1e-60 & (100,0) \\
			& $z_{0.9}$ & -2e-2 (1e-2) & 4e-2 (2e-2) & 4e-58 & (100,0) & -8e-2 (4e-2) & 5e-2 (4e-2) & 1e-65 & (100,0) \\
			\multirow{2}[1]{*}{Msd} & $0.1z_{0.9}$ & -3e-2 (2e-2) & 3e-2 (2e-2) & 9e-55 & (100,0) & -0.1 (6e-2) & 9e-4 (5e-2) & 4e-60 & (100,0) \\
			& $z_{0.9}$ & -4e-3 (3e-3) & 2e-2 (1e-2) & 3e-38 & (98,0) & -4e-3 (3e-3) & -2e-3 (1e-2) & 3e-3 & (8,92) \\
			\hline
	\end{tabular}
}%

\end{table}%

\begin{table}[htb]
	\caption{Comparison of the performance of Algorithm \ref{Omegaalgo} with plug-in and baseline approaches under GARCH setting. The values in parentheses under $\Delta_{PI}$ and $\Delta_{fun}$ are standard deviations.}
	\label{simuGARCH}%
	\centering
	\small
	\resizebox{\textwidth}{!}{
		\begin{tabular}{cc|cccc|cccc}
			\hline
			 &    & \multicolumn{4}{c|}{LB=-0.2} & \multicolumn{4}{c}{LB=-1} \\
			$F$ & $\lambda$ & $\Delta_{PI}$ & $\Delta_{fun}$ & p-value & $(n_+,n_0)$ & $\Delta_{PI}$ & $\Delta_{fun}$ & p-value & $(n_+,n_0)$ \\
			\hline
			 SR &    & -0.3 (0.2) & 0.8 (0.4) & 2e-46 & (100,0) & -0.3 (0.2) & 0.9 (0.5) & 2e-44 & (99,0) \\
			\multirow{2}[0]{*}{MV} & $0.1z_{0.9}$ & -2e-2 (2e-2) & 3e-2 (2e-2) & 3e-56 & (100,0) & -7e-2 (5e-2) & 3e-2 (4e-2) & 5e-61 & (100,0) \\
			& $z_{0.9}$ & -2e-2 (1e-2) & 3e-2 (1e-2) & 2e-63 & (100,0) & -6e-2 (3e-2) & 5e-2 (3e-2) & 6e-64 & (100,0) \\
			\multirow{2}[1]{*}{Msd} & $0.1z_{0.9}$ & -2e-2 (1e-2) & 3e-2 (1e-2) & 6e-61 & (100,0) & -9e-2 (5e-2) & 2e-2 (5e-2) & 4e-59 & (100,0) \\
			& $z_{0.9}$ & -3e-3 (3e-3) & 1e-2 (1e-2) & 2e-31 & (94,0) & -4e-3 (3e-3) & 2e-3 (1e-2) & 1e-7 & (33,65) \\
			\hline
	\end{tabular}
}%

\end{table}%

Table \ref{simuIID} shows that the baseline solution outperforms the plug-in and our proposed solutions. All settings have negative $\Delta_{PI}$ or $\Delta_{fun}$ over the baseline solution. It is consistent with the fact that the baseline solution is the exact solution of \eqref{eq1} in the IID case. Since the solution is a constant vector, it is not surprising that the plug-in method, which solves \eqref{eq1} by treating $\mathbf{w}$ as constant, performs better than our functional method. However, in the case of AR in Table \ref{simuAR} or GARCH in Table \ref{simuGARCH}, the functional approach performs much better than the plug-in approach. We have $\Delta_{PI}$ all negative for the plug-in solution. It is because baseline solution is the solution of \eqref{eq2} with exact $\bm{\mu}$ and $\bm{\Sigma}$ while plug-in solution solves the same problem \eqref{eq2} with $\bm{\mu}$ and $\bm{\Sigma}$ replaced by sample estimates. On the other hand, we have $\Delta_{fun}$ nearly all positive. The only case that $\Delta_{fun}<0$ has mean ($-0.002$) much less than the corresponding standard deviation ($0.01$). The proposed approach can outperform the oracle (baseline) approach because the solution of \eqref{eq1} is a function of past returns. We observe the similar results for the case of GARCH in Table \ref{simuGARCH}.

Although the proposed solution performs still better than plug-in solution in the settings of Msd with $\lambda=z_{0.9}$ under both AR and GARCH models, the improvement is less significant than that in other settings. It can be explained by the difficulty of variance (or second moment) estimation. It is known that $\bm{\Sigma}$ estimation is more difficult than $\bm{\mu}$ estimation. Comparing with other settings, the variance $\mathbf{w}^T\bm{\Sigma}\mathbf{w}$ has higher impact to the objective function in the Msd problem with $\lambda=z_{0.9}$ either because of larger $\lambda$ or taking square root of the variance. Hence the biases of variance estimation lead to worse performance of the functional approach in the settings of Msd with $\lambda=z_{0.9}$ than that in other settings.

\section{An empirical study}\label{sec7}

In this section, we describe an empirical study to illustrate the application of our proposed method for portfolio management. The study uses monthly stock market prices, including S{\&}P500 and its components, from January 2000 to December 2019, which are obtained from the NM FinTech database. As we can see from Figure \ref{SP}, there are several highly volatile times in the stock market in this period, including the internet bubble in the early 2000s, the recession in late 2000s, the European sovereign debt crisis that began in 2009, the Chinese stock market crash in 2015, and the Turkish currency and debt crisis in 2018. The prices are converted to log-returns by the formula $r_{i,t}=\log (P_{i,t}/P_{i,t-1})$ for the $i^{th}$ time series in the data set. We use the first ten years, from January 2000 to December 2009, data as a training data to construct portfolio weighting vector $\mathbf{w}$ by plug-in and the functional approaches as in Section \ref{sec6} and compute the return for the month of January 2010. The trading strategy is repeated in the test period, from January 2010 to December 2019 in the following manner. For a particular month $t$ in the test period, we select $m = 50$ stocks with the largest market values in the beginning of month $t$ among those that have no missing monthly prices in the previous 120 months. The log-returns of those selected stocks in the past ten years before $t$ are used as the training set $\{\mathbf{r}_{t-119},\ldots,\mathbf{r}_{t-1}\}$. The portfolios for month $t$ to be considered in the empirical study are formed from these $m$ stocks. Note that the stocks represented in $\mathbf{r}_t$ may be varying for each month $t$. Therefore our proposed approach, Algorithm \ref{Omegaalgo}, is applied for each month in the test period to update the vectors $\mathbf{a}$, $\mathbf{b}$, and $\mathbf{t}$.
\begin{figure*}[t] 
	\begin{center}
		\includegraphics[width=1\textwidth]{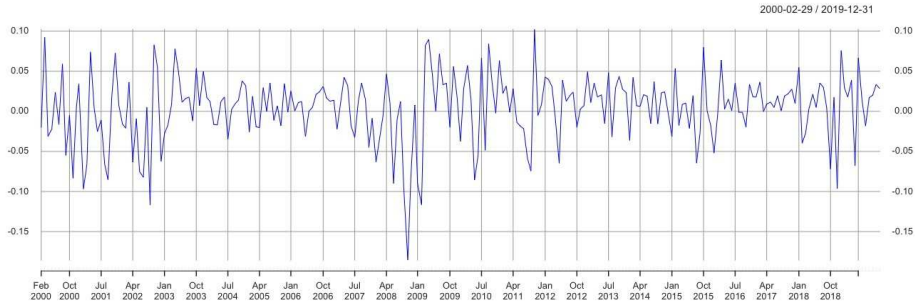}
		\caption{Monthly Returns of SP500}
		\label{SP}
	\end{center}
	
\end{figure*}

S{\&}P500 is a capitalization-weighted index whose components are weighted according to the total market value of their outstanding shares. To compare the returns of a portfolio with the returns of S{\&}P500, we consider the excess returns $e_{i,t}$ over the S{\&}P500 Index $u_t$. With the assumption that $\sum_{i=1}^p w_i=1$, we have $\sum_{i=1}^p w_i(\mathbf{S}_n)r_{i,n+1}-u_{n+1}=\sum_{i=1}^p w_i(\mathbf{S}_n)e_{i,n+1}$. We construct portfolio to maximize the information ratio $\frac{E(\mathbf{w}(\mathbf{S}_n)^T\mathbf{e}_{n+1})}{\sqrt{\text{Var}(\mathbf{w}(\mathbf{S}_n)^T\mathbf{e}_{n+1})}},$ which is of the form of problem formulation \eqref{eq8} with $r_0=0$. Here $\mathbf{e}_t=(e_{1,t},\ldots,e_{p,t})^T$, $\mathbf{S}_n=\{\mathbf{e}_1,\ldots,\mathbf{e}_n\}$, $\bm{\mu}_n(\mathbf{S}_n)=E(\mathbf{e}_{n+1}\mid \mathbf{S}_n)$ and $\mathbf{V}_n(\mathbf{S}_n)=E(\mathbf{e}_{n+1}\mathbf{e}_{n+1}^T\mid \mathbf{S}_n)$. Algorithm \ref{Omegaalgo} requires $\bm{\mu}_n$ and $\mathbf{V}_n$ to be known, and thus we need to specify time series models for the excessive returns. The top panel of Figure \ref{AMZN} gives the time series plot of excessive returns of Amazon (AMZN.O) and the bottom panel gives the autocorrelation functions of the excessive returns. They suggest that there are correlations between the lags of the excessive returns. We apply Ljung-Box test to test the autocorrelations of lags up to 12 months of the top 10 stocks that have largest market value on January 2019. The $p$-values, which are presented in Table \ref{boxtest}, show that the i.i.d.\ assumption for $e_{i,t}$ does not hold. Therefore, estimating $\bm{\mu}_n(\mathbf{S}_n)$ and $\mathbf{V}_n(\mathbf{S}_n)$ by the sample mean and second moment may not be good choices. Further time-series analysis should be conducted to get better estimators for $\bm{\mu}_n(\mathbf{S}_n)$ and $\mathbf{V}_n(\mathbf{S}_n)$. Here we simply assume the simple AR(1) model \eqref{eq10} with the returns $r_{i,t}$ replaced by excessive returns $e_{i,t}$, i.e., $e_{i,t}=\alpha_i+\beta_i e_{i,t-1}+\varepsilon_{i,t}$, and we use $\hat{\bm{\mu}}(\mathbf{S}_n)=(\hat{\alpha}_i(\mathbf{S}_n)+\hat{\beta}_i(\mathbf{S}_n)r_{i,n})_{i=1,\dots,p}$ and $\hat{\mathbf{V}}(\mathbf{S}_n)=\frac{1}{n}\sum_{t=1}^{n}\hat{\bm{\varepsilon}}_t(\mathbf{S}_n)\hat{\bm{\varepsilon}}_t(\mathbf{S}_n)^T+\hat{\bm{\mu}}(\mathbf{S}_n)\hat{\bm{\mu}}(\mathbf{S}_n)^T$ to estimate $\bm{\mu}_n(\mathbf{S}_n)$ and $\mathbf{V}_n(\mathbf{S}_n)$ as described in Section \ref{sec3}.

\begin{figure}[htbp] 
	\begin{center}
		\includegraphics[width=1\textwidth]{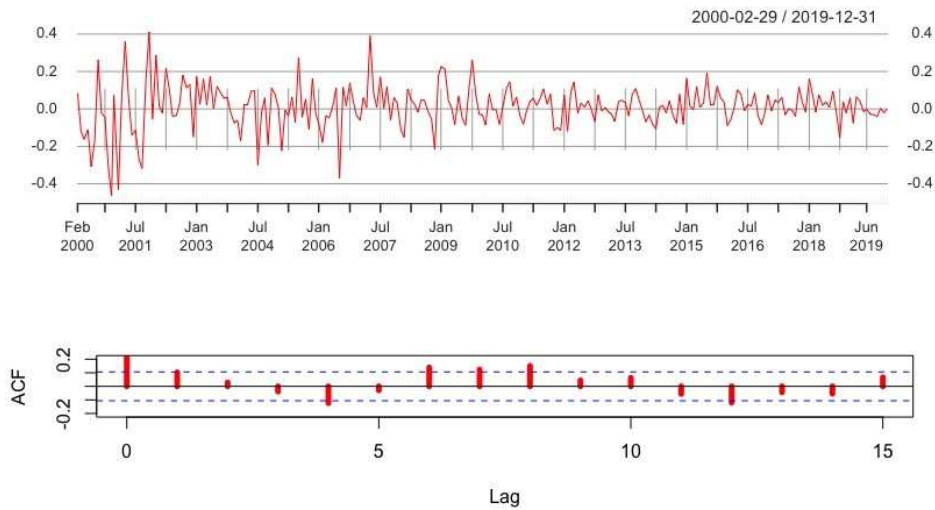}		
		\caption{Excess Returns of Amazon}
		\label{AMZN}
	\end{center}
	
\end{figure}

\begin{table}[htbp]
	\caption{Results of Ljung-Box test}
	\label{boxtest}
	\centering
	\small
	\resizebox{\textwidth}{!}{
		\begin{tabular}{c|cccccccccc}
			\hline
			 Stock & \multicolumn{1}{c}{MSFT.O} & \multicolumn{1}{c}{AAPL.O} & \multicolumn{1}{c}{AMZN.O } & \multicolumn{1}{c}{JNJ.N } & \multicolumn{1}{c}{JPM.N} & \multicolumn{1}{c}{XOM.N} & \multicolumn{1}{c}{BRK\_B.N} & \multicolumn{1}{c}{WMT.N } & \multicolumn{1}{c}{PFE.N  } & \multicolumn{1}{c}{BAC.N} \\
			\hline
			 p-values & 5.07E-01 & 5.31E-01 & 1.32E-02 & 6.05E-01 & 1.90E-05 & 6.60E-01 & 9.83E-01 & 1.19E-03 & 2.29E-01 & 6.22E-02 \\
			\hline
	\end{tabular}
}%
	
\end{table}%

Similar to the notations in Section \ref{sec6}, we let $\mathbf{w}_{PI}(\mathbf{S}_t)$ be the portfolio generated by the plug-in approach with an input $\mathbf{S}_t$, and $\mathbf{w}_{fun}(\mathbf{S}_t)$ be the portfolio generated by the the functional approach with an input $\mathbf{S}_t$ and the initial function $\mathbf{w}_0=\mathbf{w}_{PI}$. Let $e_t^{PI}=\mathbf{w}_{PI}(\mathbf{S}_{t-1})^T\mathbf{e}_t$ and $e_t^{fun}=\mathbf{w}_{fun}(\mathbf{S}_{t-1})^T\mathbf{e}_t$ be the excess returns of the portfolios $\mathbf{w}_{PI}(\mathbf{S}_t)$ and $\mathbf{w}_{fun}(\mathbf{S}_t)$ over S{\&}P500 respectively. As $t$ varies over the monthly test period from January 2000 to December 2019, we add up the realized excess returns to give the cumulative excess return $\sum_{\ell=1}^t e_{\ell}$ up to time $t$. Figures \ref{empf} gives the time series plots of the cumulative excess returns of plug-in and functional portfolios over S{\&}P500 for the constraint set $\Omega$ with $LB$ equals to -0.2 and -1 respectively. The plots show that the functional approach improves the cumulative excess return over the plug-in approach. Since we choose portfolio to maximize the information ratio, we also use the realized information ratio $R_e=\bar{e}/s_e$, where $\bar{e}$ is the sample average of the monthly excess returns and $s_e$ is the corresponding sample standard deviation, to compare the performance of $\mathbf{w}_{PI}$ and $\mathbf{w}_{fun}$.  Table \ref{empresults} shows that the information ratio $R_e^{fun}$ for functional approach is better than the corresponding ratio $R_e^{PI}$ for plug-in approach in the whole test period. When we partition the whole test period into five 2-year intervals, we still see that $R_e^{fun}$ is higher in most of the intervals.

\begin{figure}[htbp] 
	
	\begin{center}
		\includegraphics[width=1\textwidth]{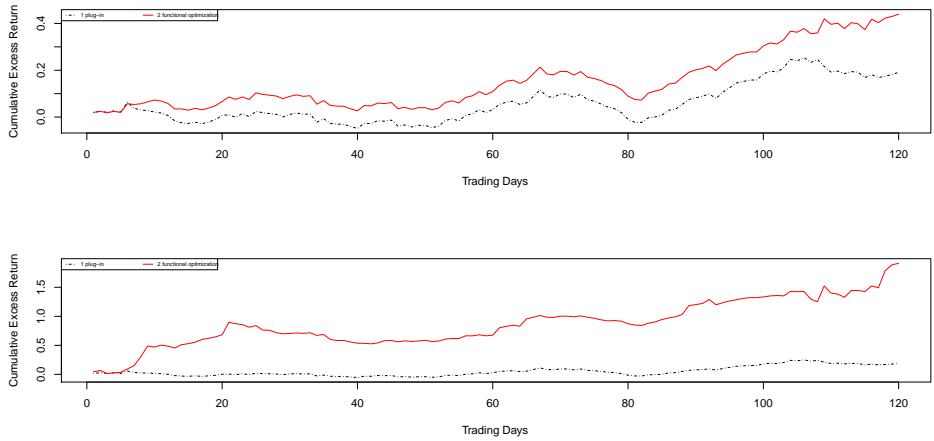}
		
		\caption{The upper panel is the time series plot of realized cumulative excess returns over S\&P500 with the constraints $\sum w_i=1$ and $w_i \geq -0.2$, the lower panel is the time series plot of realized cumulative excess returns over S\&P500 with the constraints $\sum w_i=1$ and $w_i \geq -1$.}
		\label{empf}
	\end{center}
\end{figure}

\begin{table}[htbp]
	\caption{Realized information ratios of portfolios from plug-in approach and functional approach in different intervals.}
	\label{empresults}
	\small
	\centering
	\begin{tabular}{c cc cc}
		\hline
		& \multicolumn{2}{c}{LB=-0.2} & \multicolumn{2}{c}{LB=-1} \\
		Time Interval & $R_{e}^{PI}$ & $R_{e}^{fun}$ & $R_{e}^{PI}$ & $R_{e}^{fun}$ \\
		\hline
		01.2010-12.2011 & 0.003  & 0.25  & -0.01  & 0.50  \\
		01.2012-12.2013 & -0.13  & -0.12  & -0.13  & -0.34  \\
		01.2014-12.2015 & 0.35  & 0.39  & 0.35  & 0.46  \\
		01.2016-12.2017 & 0.16  & 0.21  & 0.16  & 0.27  \\
		01.2018-12.2019 & 0.12  & 0.33  & 0.12  & 0.27  \\
		Overall & 0.10  & 0.22  & 0.10  & 0.26  \\
		\hline
	\end{tabular}%
	
\end{table}%

\section{Concluding remarks}\label{sec8}
Our work in this paper introduces a new approach to solve portfolio optimization by considering the weights as a function of past returns. We have shown that the functional approach can get the solutions that converge to the true solution of the functional optimization problem \eqref{eq1} if the underlying models for the returns are known so that the conditional mean $\bm{\mu}_n$ and conditional second moment $\mathbf{V}_n$ can be exactly evaluated. Although they are unknown in practice, our simulations and empirical results show that our approach can perform significantly better than plug-in approach when the underlying returns are correlated. Indeed our simulation results show that our functional approach can even outperform plug-in approach with true (unconditional) mean and variance. It shows the importance of our work. The majority researches try to find better estimates of unknown mean and variance, but such approaches would not be better than solving the same problem \eqref{eq1} with true mean and variance as long as the weights are considered as fixed variables. Our work shows that treating the weights as functions of past values can lead to a better portfolio. 

The functional approach proposed in this paper can be further improved in three directions. First, and the most important, conducting time-series analysis to choose a good model for the targeted returns. We can only get meaningful updates in our algorithms when the expected values can well estimated; Second, choosing a good initial function $\mathbf{w}_0$. While the expected values cannot be estimated very well, we should only run a few iterations in our algorithms. Hence, the choice of initial function is important. We use sample mean and sample second moment for the plug-in solution as the initial function in the entire paper, but obviously it can be replaced by any other choices, e.g., plug-in approach with the variance estimated based on multifactor models. Third, considering other optimization methods. We consider gradient ascent type optimization in our algorithms. It is possible to use other optimization methods to make the algorithms more efficient, e.g., stochastic gradient ascent for dynamic portfolio optimization.

\appendix
\section{Proof of Theorem $\ref{convergence}$}
 To simplify notations, in this proof, we denote $\langle\mathbf{x},\,\mathbf{y}\rangle=E(\mathbf{x}(\mathbf{S}_n)^T\mathbf{y}(\mathbf{S}_n))$ and $||\mathbf{x}||^2=\langle \mathbf{x},\mathbf{x}\rangle$. By the assumptions in Theorem 1, we assume $\mathbf{w},\mathbf{w}+\bm{\delta} \in B(\varepsilon)$.
 \begin{prop}\label{prop1}
 	Let $\tilde{G}(\mathbf{w})=G(\mathbf{w})+\frac{m}{2}||\mathbf{w}||^2$ and $\tilde{D}_1(\mathbf{w})=D_1(\mathbf{w})+m\mathbf{w}$ under the assumption A1 in Theorem 1, we have 
 	\begin{itemize}
 		\item [($a$)] $\tilde{G}(\mathbf{w}+\bm{\delta})\leq \tilde{G}(\mathbf{w})+\langle\tilde{D}_1(\mathbf{w}),\,\bm{\delta} \rangle-\frac{1}{2(M-m)}||\tilde{D}_1(\mathbf{w}+\bm{\delta})-\tilde{D}_1(\mathbf{w})||^2$
 		\item [($b$)] $\langle\tilde{D}_1(\mathbf{w}+\bm{\delta})-\tilde{D}_1(\mathbf{w}),\bm{\delta}\rangle \leq -\frac{1}{M-m}||\tilde{D}_1(\mathbf{w}+\bm{\delta})-\tilde{D}_1(\mathbf{w})||^2$
 	\end{itemize}
 \end{prop}
 
 \begin{proof}[Proof of Proposition \ref{prop1}]
 	
 	\begin{itemize}
 		\item [($a$)] It is straight forward to see that assumption A1 implies 
 		$$\tilde{G}(\mathbf{w}+\bm{\delta}) \leq \tilde{G}(\mathbf{w})+\langle\tilde{D}_1(\mathbf{w}),\bm{\delta}\rangle ,\quad
 		\tilde{G}(\mathbf{w}+\bm{\delta}) \geq \tilde{G}(\mathbf{w})+\langle\tilde{D}_1(\mathbf{w}),\bm{\delta}\rangle - \frac{M-m}{2}||\bm{\delta}||^2$$
 		Therefore, let $\bm{\delta}^{\prime}=\bm{\delta}+\frac{1}{M-m}(\tilde{D}_1(\mathbf{w}+\bm{\delta})-\tilde{D}_1(\mathbf{w}))$,
 		\begin{align*}
 		\begin{aligned}
 		&\tilde{G}(\mathbf{w}+\bm{\delta})-\tilde{G}(\mathbf{w})\\ =& \tilde{G}(\mathbf{w}+\bm{\delta})-\tilde{G}(\mathbf{w}+\bm{\delta}^{\prime})+\tilde{G}(\mathbf{w}+\bm{\delta}^{\prime})-\tilde{G}(\mathbf{w})\\
 		\leq& -\langle\tilde{D}_1(\mathbf{w}+\bm{\delta}),\bm{\delta}^{\prime}-\bm{\delta}\rangle + \frac{M-m}{2}||\bm{\delta}^{\prime}-\bm{\delta}||^2+\langle\tilde{D}_1(\mathbf{w}),\bm{\delta}^{\prime}\rangle \\
 		=&-\frac{1}{M-m} \langle\tilde{D}_1(\mathbf{w}+\bm{\delta}),\tilde{D}_1(\mathbf{w}+\bm{\delta})-\tilde{D}_1(\mathbf{w})\rangle +\frac{1}{2(M-m)}||\tilde{D}_1(\mathbf{w}+\bm{\delta})\\&-\tilde{D}_1(\mathbf{w})||^2
 		+\langle\tilde{D}_1(\mathbf{w}),\bm{\delta}\rangle +\frac{1}{M-m}\langle\tilde{D}_1(\mathbf{w}),\tilde{D}_1(\mathbf{w}+\bm{\delta})-\tilde{D}_1(\mathbf{w})\rangle\\
 		=&\langle\tilde{D}_1(\mathbf{w}),\bm{\delta}\rangle -\frac{1}{2(M-m)}||\tilde{D}_1(\mathbf{w}+\bm{\delta})-\tilde{D}_1(\mathbf{w})||^2
 		\end{aligned}
 		\end{align*}
 		\item [($b$)] By exchanging the role of $\mathbf{w}$ and $\mathbf{w}+\bm{\delta}$ in ($a$), we have
 		$$\tilde{G}(\mathbf{w})\leq\tilde{G}(\mathbf{w}+\bm{\delta})+\langle\tilde{D}_1(\mathbf{w}+\bm{\delta}),-\bm{\delta}\rangle-\frac{1}{2(M-m)}||\tilde{D}_1(\mathbf{w})-\tilde{D}_1(\mathbf{w}+\bm{\delta})||^2$$
 		By summing up this inequality with the 
 		one in ($a$), we get the inequality ($b$).
 	\end{itemize}	
 \end{proof}
 Substitute $\tilde{D}_1(\mathbf{w})=D_1(\mathbf{w})+m\mathbf{w}$ in Proposition 1 ($b$), we have
 \begin{align*}
 \begin{aligned}
 &\langle D_1(\mathbf{w}+\bm{\delta})-D_1(\mathbf{w})+m\bm{\delta},\bm{\delta}\rangle \leq -\frac{1}{M-m}||D_1(\mathbf{w}+\bm{\delta})-D_1(\mathbf{w})+m\bm{\delta}||^2\\\Leftrightarrow & (M-m)\langle D_1(\mathbf{w}+\bm{\delta})-D_1(\mathbf{w}),\bm{\delta}\rangle + m(M-m)||\bm{\delta}||^2\\&\leq -(||D_1(\mathbf{w}+\bm{\delta})-D_1(\mathbf{w})||^2+2m\langle D_1(\mathbf{w}+\bm{\delta})-D_1(\mathbf{w}),\bm{\delta}\rangle+m^2||\bm{\delta}||^2)\\
 \Leftrightarrow&(m+M)\langle D_1(\mathbf{w}+\bm{\delta})-D_1(\mathbf{w}),\bm{\delta}\rangle \leq -mM||\bm{\delta}||^2-||D_1(\mathbf{w}+\bm{\delta})-D_1(\mathbf{w})||^2\\
 \Leftrightarrow&\langle D_1(\mathbf{w}+\bm{\delta})-D_1(\mathbf{w}),\bm{\delta}\rangle \leq-(\frac{mM}{m+M}||\bm{\delta}||^2+\frac{1}{m+M}||D_1(\mathbf{w}+\bm{\delta})-D_1(\mathbf{w})||^2)
 \end{aligned}
 \end{align*} Therefore, \begin{align*}
 \begin{aligned}
 &||\mathbf{w}_{k+1}-\mathbf{w^*}||^2\\
 =&||\mathbf{w}_k+\frac{2}{M+m}\mathbf{P}D_1(\mathbf{w}_k)-\mathbf{w^*}||^2\\
 =&||\mathbf{P}(\mathbf{w}_k-\mathbf{w^*})+\frac{2}{M+m}\mathbf{P}(D_1(\mathbf{w}_k)-D_1(\mathbf{w^*}))||^2\\
 =&||\mathbf{P}(\mathbf{w}_k-\mathbf{w^*}+\frac{2}{M+m}(D_1(\mathbf{w}_k)-D_1(\mathbf{w^*})))||^2\\
 \leq& ||\mathbf{w}_k-\mathbf{w^*}+\frac{2}{M+m}(D_1(\mathbf{w}_k)-D_1(\mathbf{w^*}))||^2\\
 =&||\mathbf{w}_k-\mathbf{w^*}||^2+\frac{4}{M+m}\langle D_1(\mathbf{w}_k)-D_1(\mathbf{w^*}),\mathbf{w}_k-\mathbf{w^*}\rangle+\frac{4}{(M+m)^2}||D_1(\mathbf{w}_k)-D_1(\mathbf{w^*})||^2\\
 \leq& ||\mathbf{w}_k-\mathbf{w^*}||^2-\frac{4mM}{(m+M)^4}||\mathbf{w}_k-\mathbf{w^*}||^2-\frac{4}{(M+m)^2}||D_1(\mathbf{w}_k)-D_1(\mathbf{w^*})||^2\\ &+\frac{4}{(M+m)^2}||D_1(\mathbf{w}_k)-D_1(\mathbf{w^*})||^2\\
 =&(\frac{M-m}{M+m})^2||\mathbf{w}_k-\mathbf{w}||^2
 \end{aligned}
 \end{align*} 
 $\therefore \quad ||\mathbf{w}_k-\mathbf{w^*}||\leq\frac{M-m}{M+m}||\mathbf{w}_{k-1}-\mathbf{w^*}||\leq(\frac{M-m}{M+m})^k||\mathbf{w}_0-\mathbf{w^*}||$

\section{Proof of Theorem \ref{thm2}}
 Note that Proposition \ref{prop1} and the arguments in the first part of the proof of Theorem \ref{convergence} still hold. So we just need to show inequality \eqref{eq17} directly. With $t_k = \frac{2}{M+m}$, we have
 \begin{align*}
 \begin{aligned}
 ||\mathbf{w}_{k+1}-\mathbf{w^*}||^2 &=||P_\Omega(\mathbf{w}_k+tD_1(\mathbf{w}_k))-P_\Omega(\mathbf{w^*}+tD_1(\mathbf{w^*}))||^2\\ 
 &\leq||\mathbf{w}_k-\mathbf{w^*}+\frac{2}{M+m}(D_1(\mathbf{w}_k-D_1(\mathbf{w^*})))||^2\\
 &\leq(\frac{M-m}{M+m})^2||\mathbf{w}_k-\mathbf{w^*}||^2
 \end{aligned}
 \end{align*}


\end{document}